\documentclass[review]{elsarticle}

\usepackage{graphicx}
\usepackage{amsfonts,amsmath,amssymb,amsthm}
\usepackage{microtype}
\usepackage{array}
\usepackage{booktabs}
\usepackage{longtable}
\usepackage{multirow}
\usepackage{makecell}

\usepackage{tikz}
\usepackage{pgfplots}
\pgfplotsset{compat=1.18}
\usetikzlibrary{arrows.meta,calc,positioning}

\usepackage{xcolor}
\definecolor{darkblue}{RGB}{0,70,140}

\usepackage{hyperref}
\hypersetup{
  colorlinks=true,
  linkcolor=darkblue,
  citecolor=darkblue,
  urlcolor=darkblue
}

\journal{Insurance: Mathematics and Economics}

\theoremstyle{plain}
\newtheorem{theorem}{Theorem}[section]
\newtheorem{proposition}[theorem]{Proposition}
\newtheorem{lemma}[theorem]{Lemma}
\newtheorem{corollary}[theorem]{Corollary}

\theoremstyle{definition}
\newtheorem{definition}[theorem]{Definition}

\theoremstyle{remark}

\theoremstyle{plain}
\newtheorem{assumption}{Assumption}

\newcommand{\Var}{\operatorname{Var}}
\newcommand{\Cov}{\operatorname{Cov}}
\newcommand{\SMD}{\mathrm{SMD}}
\newcommand{\FMD}{\mathrm{FMD}}
\newcommand{\E}{\mathbb{E}}

\newcommand{\bX}{\boldsymbol{X}}
\newcommand{\bx}{\boldsymbol{x}}
\newcommand{\bbeta}{\boldsymbol{\beta}}
\newcommand{\bgamma}{\boldsymbol{\gamma}}
\newcommand{\btheta}{\boldsymbol{\vartheta}}
\newcommand{\bmu}{\boldsymbol{\mu}}
\newcommand{\bSigma}{\boldsymbol{\Sigma}}

\begin{document}

\begin{frontmatter}

\title{Direct and Indirect Discrimination\\in Generalized Linear Models}

\author[milliman]{Bertille Tierny}
\ead{bertille.tierny@milliman.com}

\author[uqam,kyoto]{Arthur Charpentier}
\ead{charpentier.arthur@uqam.ca}

\author[milliman,lyon]{François Hu}
\ead{francois.hu@milliman.com}

\address[milliman]{Milliman R\&D, Paris, France}
\address[uqam]{Université du Québec à Montréal, Canada}
\address[kyoto]{Kyoto University, Japan}
\address[lyon]{LSAF, Université Claude Bernard Lyon 1, France}

\begin{abstract}
Generalized linear models are central to actuarial modelling of binary risk, claim frequency, utilization, and cost-related outcomes. Yet fairness diagnostics often rely on linear-model intuitions, although GLM predictions are obtained by transporting a latent score through a nonlinear inverse link. We develop a moment-based decomposition framework for diagnosing group disparities in fitted GLM predictions. In an exact linear-Gaussian benchmark, the Wasserstein barycentric criterion for distributional demographic-parity violation reduces to a two-moment criterion and decomposes into direct mean, indirect mean, interaction, and structural components. For GLMs, we distinguish the empirical output-scale criterion $U_2(f)$, a within-group proxy $\widetilde U_2(f)$, and a leading decomposition $D_1(f)$. This leading term preserves the four linear channels and adds two curvature components induced by the inverse link: curvature coupling and curvature amplification. We derive explicit formulas for logistic, Poisson, and Tweedie specifications and illustrate the diagnostic on medical-expenditure survey data. The framework is not a legal test of discrimination, nor a full characterization of distributional parity outside the linear-Gaussian case. It is a tractable actuarial diagnostic for identifying whether fitted prediction disparities arise from explicit sensitive effects, proxy-mediated covariate profiles, covariance-structure differences, or nonlinear link effects.
\end{abstract}

\begin{keyword}
Discrimination \sep Generalized linear models \sep Insurance pricing \sep Fairness
\MSC[2020] 62J12 \sep 91G70
\end{keyword}

\end{frontmatter}

\section{Introduction and Motivation}
\label{sec:introduction}

Generalized linear models (GLMs) remain among the main statistical workhorses of actuarial modelling. In insurance, health applications, and risk classification, the quantities of interest are rarely Gaussian on their raw scale.
Claim occurrence is binary, claim frequency is discrete, and utilization or loss-related outcomes are often skewed, zero-inflated, and heterogeneous.
GLMs address these features by combining interpretable predictors, response distributions adapted to non-Gaussian outcomes, and link functions that map a latent score to a prediction on the scale used for actuarial decisions \cite{mccullagh1989generalized,de2008generalized,kaas2008modern,frees2009regression,ohlsson2010non,frees2014predictive,frees2016predictive,dunn2018generalized,denuit2019generalized}. Logistic regression, Poisson regression, and Tweedie-type models are therefore standard tools for modelling binary risk, frequency, utilization, and cost-related outcomes. Generalized additive models (GAMs) provide a closely related extension when nonlinear effects must remain interpretable \cite{hastie1990generalized,wood2017generalized}.

Fairness questions are particularly delicate in insurance because classification and segmentation are intrinsic to actuarial practice. Risk-based pricing and risk scoring differentiate policyholders according to predictive variables, but some of these variables may also be correlated with protected or socially sensitive characteristics. This creates a tension between actuarial relevance, legal acceptability, and ethical legitimacy \cite{baumann2023fairnessrisk,lindholm2022discrimination,cote2025fairprice,charpentier2024insurance}. Removing an explicitly sensitive variable from a tariff is therefore not enough. Ordinary rating factors may still carry sensitive information, and disparities may persist through differences in covariate profiles, in covariance structures, or in the way a nonlinear link transforms latent scores into predictions.

This paper develops a decomposition framework for diagnosing such disparities in actuarial GLMs. The objective is not to propose a legal test of discrimination, nor to replace model validation or actuarial judgement. The objective is narrower and operational: given a fitted model, identify which statistical channels explain differences in fitted predictions across sensitive groups. This is relevant for pricing, segmentation, compliance reviews, model audits, and governance discussions, because different channels call for different responses. An explicit sensitive coefficient, a proxy-mediated covariate profile effect, a structural dispersion effect, and a link-induced curvature effect do not have the same interpretation.

Throughout the paper, the terms direct and indirect are used in this statistical diagnostic sense. A direct component refers to the explicit fitted contribution of the sensitive attribute in the score. An indirect component refers to disparities transmitted through the distribution of non-sensitive covariates. These terms are related to familiar discussions about explicit and proxy-driven disparities, but they do not by themselves establish a legal finding of discrimination.

The starting point is an exact linear benchmark. Under the identity link, the latent score and the fitted prediction coincide. If the group-conditional score distributions are Gaussian, the Wasserstein barycentric measure of distributional demographic-parity violation reduces exactly to a two-moment criterion. This criterion decomposes into four channels: a direct mean component, an indirect mean component, an interaction component, and an indirect structural component. The interaction term records whether the direct and indirect mean channels reinforce or offset one another. The structural term captures differences in within-group score dispersion induced by differences in the covariance structure of the covariates.

Actuarial GLMs require an additional step. Their predictions are read on the mean scale,
\[
  f(\bx,s)=h\{\eta(\bx,s)\},
\]
where $h$ is the inverse link and $\eta$ is the latent score. A disparity that
is additive on the latent scale may be compressed, amplified, or reshaped after
the nonlinear transformation. For this reason, the linear decomposition cannot
be transferred mechanically to GLMs. The link function itself becomes part of
the diagnostic.

The GLM extension separates three objects. The first is the empirical
output-scale criterion
\[
  U_2(f),
\]
which measures between-group differences in the conditional means and
standard deviations of fitted predictions. The second is a within-group proxy
\[
  \widetilde U_2(f),
\]
obtained by transporting the first two latent-score moments through a first-order
delta approximation. The third is a leading decomposition
\[
  D_1(f),
\]
obtained by expanding the proxy criterion across sensitive groups around the
global latent mean. This yields the identity
\[
  U_2(f)
  =
  D_1(f)
  +
  \Delta_{\mathrm{within}}(f)
  +
  \Delta_{\mathrm{across}}(f),
\]
where the two error terms make the approximation explicit. The leading decomposition $D_1(f)$ preserves the four linear channels and adds two curvature terms induced by the inverse link: a curvature coupling term, which links group-wise latent means and score dispersions, and a curvature amplification term, which converts latent mean heterogeneity into additional prediction-scale dispersion.

This distinction is important for empirical work. Outside the linear-Gaussian benchmark, the proposed criterion is not a full characterization of distributional demographic parity. It tracks differences in the first two conditional moments of the fitted predictions, not equality of the full conditional prediction laws. Moreover, the GLM decomposition is first-order.
Its accuracy depends on within-group score dispersion, separation of group-wise latent means, and curvature of the inverse link over the fitted score range. For this reason, the empirical analysis reports the exact criterion $U_2(f)$, the proxy $\widetilde U_2(f)$, the leading decomposition $D_1(f)$, and the approximation gaps.

The paper also studies a latent moment-aligning post-processing. This correction is applied to the latent score, not directly to the prediction scale.
It aligns the first two conditional moments of the fitted score across sensitive groups and therefore eliminates the first-order proxy disparity. In nonlinear GLMs, however, residual output-scale disparities may remain because the inverse link transforms higher-order differences in the conditional score distributions. The correction is therefore used as a diagnostic and comparative device, not as a normative definition of a fair tariff.

The empirical illustrations use the Medical Expenditure Panel Survey Household Component (MEPS-HC) data. We consider two sensitive comparisons: Hispanic versus non-Hispanic individuals, and uninsured versus insured individuals. The first comparison reflects an ethnicity-related disparity, while
the second reflects a coverage-related disparity linked to access to care,
utilization, and socioeconomic vulnerability. The same empirical design is used
across model classes to show how the diagnostic changes when the modelling
target and link function change.

The remainder of the paper is organized as follows.
Section~\ref{sec:lm-decomposition} presents the exact linear benchmark, the
four-way decomposition, and the latent moment-aligning correction in the
identity-link case. Section~\ref{sec:glm-decomposition} develops the general
GLM decomposition, distinguishes $U_2(f)$, $\widetilde U_2(f)$, and
$D_1(f)$, formalizes the approximation errors, and extends the construction
to GAMs. Section~\ref{sec:logistic} specializes the framework to logistic
regression and illustrates the decomposition on hospitalization risk.
Section~\ref{sec:poisson} treats Poisson regression and applies the diagnostic
to office-visit frequency. Section~\ref{sec:tweedie} studies Tweedie GLMs,
distinguishing canonical and log-link specifications; the empirical application
uses a log-link variance-power model for office-visit utilization, followed by a discussion of aggregate losses and pure-premium models.
Section~\ref{sec:conclusion} discusses the scope, limitations, and actuarial interpretation of the proposed diagnostic. The appendix reports the empirical design and the computation of the diagnostic quantities.
\section{Linear Benchmark}
\label{sec:lm-decomposition}

The generalized linear model decompositions developed later in the paper are
built from an exact linear benchmark. Under the identity link, the latent score
and the prediction coincide, so disparities in predictions can be decomposed
without approximation. This benchmark follows the decomposition logic developed for linear models in
\cite{tierny2025decomposing}, while adapting it here to the notation and
actuarial interpretation needed for the GLM extension. This section
introduces the notation, connects the construction to Oaxaca--Blinder-type
decompositions, derives the exact four-way decomposition, and illustrates the
resulting diagnostic on medical expenditure data.

Throughout the paper, the terms direct and indirect are used in a statistical
diagnostic sense. A direct component refers to the explicit fitted contribution
of the sensitive attribute in the score. An indirect component refers to
disparities transmitted through the distribution of the non-sensitive covariates.
These terms do not, by themselves, constitute a legal characterization of
discrimination.

\subsection{Prediction disparities under a single fitted model}
\label{subsec:lm-prediction-disparities}

Classical Kitagawa--Oaxaca--Blinder decompositions explain an observed outcome
gap between groups by separating differences in characteristics from differences
in coefficients \cite{kitagawa1955components,oaxaca1973male,blinder1973wage};
see also \cite{jann2008blinder,fortin2011decomposition}. Our object is
different. We do not compare two group-specific outcome regressions. We study
the disparities generated by a single fitted predictor when it is applied to
different sensitive groups. The analogy with Oaxaca--Blinder is therefore
structural: the aim is to decompose a prediction disparity into interpretable
channels, not to decompose an observed outcome gap. Related nonlinear and
distributional extensions exist \cite{fairlie2005extension}. In the fairness
literature, linear models under demographic-parity constraints have also been
studied from an optimization and minimax perspective
\cite{fukuchi2023demographic}. The present paper takes a complementary route:
it uses the linear case as an exact diagnostic benchmark for decomposing
prediction disparities, before transporting the same logic to actuarial GLMs.

Let $\bX\in\mathbb R^k$ denote the vector of non-sensitive covariates, $S\in[M]:=\{1,\ldots,M\}$ a discrete sensitive attribute, and $Y$ the response.
We write
\[
  \pi_s:=\mathbb P(S=s),
  \qquad s\in[M],
\]
and assume $\pi_s>0$. For each group $s$, define
\[
  \bmu^{(s)}:=\E[\bX\mid S=s],
  \qquad
  \bSigma^{(s)}:=\Cov(\bX\mid S=s).
\]
For any collection $(z_s)_{s\in[M]}$, define
\[
  \E_S[z_S]:=\sum_{s=1}^M\pi_s z_s,
  \qquad
  \Var_S(z_S):=\sum_{s=1}^M\pi_s
  \bigl(z_s-\E_S[z_S]\bigr)^2,
\]
and
\[
  \Cov_S(z_S,z'_S)
  :=
  \sum_{s=1}^M\pi_s
  \bigl(z_s-\E_S[z_S]\bigr)
  \bigl(z'_s-\E_S[z'_S]\bigr).
\]

Consider the linear predictor
\[
  f_{\mathrm{lin}}(\bx,s)
  =
  \langle \bx,\bbeta\rangle+\gamma_s+\beta_0,
\]
where $\bbeta\in\mathbb R^k$, $\beta_0\in\mathbb R$, and $\gamma_s$ is the group-specific fitted contribution of the sensitive attribute. The usual identifiability constraint on the collection $(\gamma_s)_{s\in[M]}$ is immaterial for the decomposition below, because the terms involve only between-group variances and covariances.

The group-wise moments of the fitted prediction are
\begin{equation}
\label{eq:vs-bs-linear}
\begin{cases}
m_s
:=
\E[f_{\mathrm{lin}}(\bX,S)\mid S=s]
=
\langle \bmu^{(s)},\bbeta\rangle+\gamma_s+\beta_0,
\\[0.3em]
v_s
:=
\Var(f_{\mathrm{lin}}(\bX,S)\mid S=s)
=
\bbeta^\top\bSigma^{(s)}\bbeta.
\end{cases}
\end{equation}
We set
\[
  b_s:=\sqrt{v_s},
  \qquad
  M:=m_S,
  \qquad
  B:=b_S,
  \qquad
  \bar m:=\E_S[M],
  \qquad
  \bar b:=\E_S[B].
\]

\subsection{Distributional parity and Gaussian reduction}
\label{subsec:lm-dp}

Distributional demographic parity requires the conditional prediction law to be the same across sensitive groups:
\[
  \nu_{f_{\mathrm{lin}}\mid S=s}
  =
  \nu_{f_{\mathrm{lin}}\mid S=t},
  \qquad s,t\in[M],
\]
where $\nu_{f_{\mathrm{lin}}\mid S=s}$ denotes the law of $f_{\mathrm{lin}}(\bX,S)$ conditional on $S=s$. We measure departures from this condition by the Wasserstein barycentric criterion
\begin{equation}
\label{eq:w2-unfairness-linear}
  U_{W_2}(f_{\mathrm{lin}})
  :=
  \min_{\nu\in\mathcal P_2(\mathbb R)}
  \sum_{s=1}^M
  \pi_s W_2^2\bigl(\nu_{f_{\mathrm{lin}}\mid S=s},\nu\bigr),
\end{equation}
where $\mathcal P_2(\mathbb R)$ is the set of probability measures on $\mathbb R$ with finite second moment.

\begin{assumption}[Gaussian group-conditional score distributions]
\label{ass:gaussian-linear}
For each $s\in[M]$,
\[
  f_{\mathrm{lin}}(\bX,S)\mid S=s
  \sim
  \mathcal N(m_s,b_s^2).
\]
A sufficient condition is that $\bX\mid S=s$ be Gaussian.
\end{assumption}

Under Assumption~\ref{ass:gaussian-linear}, the Wasserstein criterion has an exact two-moment form. This is the only point in the paper where the moment-based criterion coincides exactly with a full distributional parity criterion.

\begin{definition}[Linear moment-based disparity]
\label{def:U2-linear}
Define
\[
  U_2(f_{\mathrm{lin}})
  :=
  \Var_S(m_S)+\Var_S(b_S).
\]
We write
\[
  U_2(f_{\mathrm{lin}})
  =
  \underbrace{\Var_S(m_S)}_{\mathrm{FMD}}
  +
  \underbrace{\Var_S(b_S)}_{\mathrm{SMD}},
\]
where $\mathrm{FMD}$ denotes first-moment disparity and $\mathrm{SMD}$
second-moment disparity.
\end{definition}

\begin{proposition}[Exact Gaussian reduction]
\label{prop:gaussian-reduction-linear}
Under Assumption~\ref{ass:gaussian-linear},
\[
  U_{W_2}(f_{\mathrm{lin}})
  =
  U_2(f_{\mathrm{lin}})
  =
  \Var_S(m_S)+\Var_S(b_S).
\]
\end{proposition}

\begin{proof}
Let $Q_s$ be the quantile function of
$f_{\mathrm{lin}}(\bX,S)\mid S=s$. Under
Assumption~\ref{ass:gaussian-linear},
\[
  Q_s(u)=m_s+b_s\Phi^{-1}(u),
  \qquad 0<u<1,
\]
where $\Phi$ is the standard normal distribution function. In one dimension, the squared Wasserstein distance can be written as the squared $L^2$ distance between quantile functions, and the Wasserstein barycenter has quantile function
\[
  Q_\star(u)
  =
  \E_S[Q_S(u)]
  =
  \bar m+\bar b\,\Phi^{-1}(u).
\]
Hence
\begin{align*}
  U_{W_2}(f_{\mathrm{lin}})
  &=
  \int_0^1
  \Var_S\bigl(Q_S(u)\bigr)\,du
  \\
  &=
  \int_0^1
  \Var_S\bigl(m_S+b_S\Phi^{-1}(u)\bigr)\,du.
\end{align*}
Expanding the variance gives
\[
  \Var_S(m_S)
  +
  2\Phi^{-1}(u)\Cov_S(m_S,b_S)
  +
  \{\Phi^{-1}(u)\}^2\Var_S(b_S).
\]
Since
\[
  \int_0^1 \Phi^{-1}(u)\,du=0,
  \qquad
  \int_0^1 \{\Phi^{-1}(u)\}^2\,du=1,
\]
we obtain
\[
  U_{W_2}(f_{\mathrm{lin}})
  =
  \Var_S(m_S)+\Var_S(b_S).
\]
\end{proof}

Thus, in the linear-Gaussian benchmark, $U_2$ is not merely a surrogate for \eqref{eq:w2-unfairness-linear}. It is exactly the Wasserstein barycentric measure of distributional demographic-parity violation.

\subsection{Exact four-way decomposition}
\label{subsec:lm-exact-decomposition}

The first-moment path decomposes directly. Since
\[
  M
  =
  m_S
  =
  \langle \bmu^{(S)},\bbeta\rangle+\gamma_S+\beta_0,
\]
and $\beta_0$ is constant across groups,
\begin{equation}
\label{eq:mean-path-decomposition-linear}
\Var_S(M)
=
\Var_S(\gamma_S)
+
\Var_S\bigl(\langle \bmu^{(S)},\bbeta\rangle\bigr)
+
2\,\Cov_S\bigl(\gamma_S,\langle \bmu^{(S)},\bbeta\rangle\bigr).
\end{equation}

\begin{theorem}[Exact four-way decomposition in linear models]
\label{thm:linear-four-way}
The linear moment-based disparity satisfies
\begin{align}
U_2(f_{\mathrm{lin}})
&=
\underbrace{\Var_S(\gamma_S)}_{\textnormal{direct mean}}
+
\underbrace{\Var_S\bigl(\langle \bmu^{(S)},\bbeta\rangle\bigr)}
_{\textnormal{indirect mean}}
\notag\\
&\quad+
\underbrace{2\,\Cov_S\bigl(\gamma_S,
\langle \bmu^{(S)},\bbeta\rangle\bigr)}_{\textnormal{interaction}}
+
\underbrace{\Var_S(B)}_{\textnormal{indirect structural}}.
\label{eq:linear-four-way-theorem}
\end{align}
Equivalently,
\[
U_2(f_{\mathrm{lin}})
=
\Var_S(\gamma_S)
+
\Var_S\bigl(\langle \bmu^{(S)},\bbeta\rangle\bigr)
+
2\,\Cov_S\bigl(\gamma_S,\langle \bmu^{(S)},\bbeta\rangle\bigr)
+
\Var_S\left(\sqrt{\bbeta^\top\bSigma^{(S)}\bbeta}\right).
\]
\end{theorem}

\begin{proof}
By Definition~\ref{def:U2-linear},
\[
  U_2(f_{\mathrm{lin}})
  =
  \Var_S(m_S)+\Var_S(b_S).
\]
Apply \eqref{eq:mean-path-decomposition-linear} to the first term. The second term is
\[
  \Var_S(b_S)
  =
  \Var_S\left(\sqrt{\bbeta^\top\bSigma^{(S)}\bbeta}\right),
\]
by \eqref{eq:vs-bs-linear}. Combining the two identities gives the result.
\end{proof}

The four terms have distinct diagnostic meanings. The direct mean component is generated by the explicit group contribution $\gamma_S$. The indirect mean component is generated by differences in average non-sensitive covariate profiles, projected through the fitted coefficient vector $\bbeta$. The interaction term records whether the direct and indirect mean channels reinforce or offset each other. A positive interaction means that groups with a larger direct fitted contribution also tend to have larger covariate-induced fitted scores. A negative interaction means that the two channels partially offset each other; it should not be interpreted as a normative correction, but only as a statistical compensation in the fitted score. Finally, the indirect structural component comes from differences in the within-group covariance structure of the covariates entering the score. Hence, two groups may have similar average predictions but different predictive dispersion.

Figure~\ref{fig:schema:LM} summarizes the two mean channels and the structural dispersion channel in the linear benchmark.

\begin{figure}[!htbp]
    \centering
    \begin{tikzpicture}[
  >=Latex,
  scale=.9,
  every node/.style={font=\small},
  tealA/.style={color=teal!65!black},
  orangeB/.style={color=orange!85!black}
]

\def\betap{0.70}
\def\betaZero{0.45}
\def\gammaA{0.00}
\def\gammaB{1.25}

\def\xA{1.8}
\def\xB{5.0}

\pgfmathsetmacro{\mA}{\betaZero+\gammaA+\betap*\xA}
\pgfmathsetmacro{\mB}{\betaZero+\gammaB+\betap*\xB}
\pgfmathsetmacro{\mAatB}{\betaZero+\gammaA+\betap*\xB}

\def\bA{0.35}
\def\bB{0.65}

\def\xd{6.25}
\def\denswidth{0.45}

\draw[->] (0,0) -- (7.2,0) node[right] {$x$};
\draw[->] (0,0) -- (0,6.4) node[above] {$f_{\rm lin}(x,s)$};

\draw[teal!45, dashed] (\xA,0) -- (\xA,6.1);
\draw[orange!45, dashed] (\xB,0) -- (\xB,6.1);

\draw[teal!45, dashed] (0,\mA) -- (6.9,\mA);
\draw[orange!45, dashed] (0,\mB) -- (6.9,\mB);

\node[teal!65!black, below] at (\xA,0) {$\bar x_A$};
\node[orange!85!black, below] at (\xB,0) {$\bar x_B$};

\node[teal!65!black, left] at (0,\mA) {$m_A$};
\node[orange!85!black, left] at (0,\mB) {$m_B$};

\draw[teal!65!black, very thick]
  plot[domain=0.2:6.4, samples=2]
  (\x,{\betaZero+\gammaA+\betap*\x});

\draw[orange!85!black, very thick]
  plot[domain=0.2:6.4, samples=2]
  (\x,{\betaZero+\gammaB+\betap*\x});

\foreach \x/\y in {
1.1/1.0,1.3/1.2,1.4/1.35,1.5/1.5,1.6/1.25,
1.7/1.7,1.8/1.55,1.9/1.85,2.0/1.65,2.1/1.95,
2.2/2.05,2.3/1.75,2.4/2.15,2.5/2.30,1.2/0.85,
1.45/1.15,1.65/1.35,1.95/1.75,2.15/1.90
}{
  \fill[teal!55, opacity=.28] (\x,\y) circle (1.8pt);
}

\foreach \x/\y in {
4.5/4.6,4.6/4.85,4.7/5.05,4.8/4.55,4.9/5.15,
5.0/4.95,5.1/5.35,5.2/5.10,5.3/5.55,5.4/5.25,
5.5/5.65,5.6/5.35,5.7/5.85,5.8/5.70,4.65/4.75,
4.85/5.00,5.05/5.20,5.25/5.40,5.45/5.55
}{
  \fill[orange!65, opacity=.30] (\x,\y) circle (1.8pt);
}

\fill[black] (\xA,\mA) circle (2.3pt);
\fill[black] (\xB,\mAatB) circle (2.3pt);
\fill[black] (\xB,\mB) circle (2.3pt);

\draw[black, very thick, ->]
  (\xA,\mA) -- (\xB,\mAatB);

\node[teal!65!black] at (3.15,3.25) {indirect};

\draw[black, very thick, <->]
  (\xB,\mAatB) -- (\xB,\mB);

\node[orange!85!black, right] at (\xB,{(\mAatB+\mB)/2+.2}) {direct};

\draw[teal!65!black, thin]
  plot[domain={\mA-4*\bA}:{\mA+4*\bA}, samples=100, variable=\y]
  ({\xd + 2*\denswidth*exp(-((\y-\mA)^2)/(2*\bA*\bA))}, {\y});

\draw[orange!85!black, thin]
  plot[domain={\mB-4*\bB}:{\mB+4*\bB}, samples=100, variable=\y]
  ({\xd + 2*\denswidth*exp(-((\y-\mB)^2)/(2*\bB*\bB))}, {\y});

\draw[black, thick, <->]
  ({\xd+\denswidth+.10},{\mA-\bA}) --
  ({\xd+\denswidth+.10},{\mA+\bA});
\node[teal!65!black, right] at ({\xd+\denswidth+.16*3},\mA) {$b_A$};

\draw[black, thick, <->]
  ({\xd+\denswidth+.10},{\mB-\bB}) --
  ({\xd+\denswidth+.10},{\mB+\bB});
\node[orange!85!black, right] at ({\xd+\denswidth+.16*3},\mB) {$b_B$};

\node[font=\bfseries] at (3.6,7.25) {Linear benchmark};
\node[font=\small] at (3.6,6.75)
  {$f_{\rm lin}(x,s)=\beta_0+\gamma_s+\beta x$};

\end{tikzpicture}
    \caption{Schematic linear decomposition. The group-average score gap is read through an indirect mean channel, driven by differences in average non-sensitive covariates, and a direct mean channel, driven by the explicit group effect. The small Gaussian shapes summarize group-specific predictive dispersion through $b_A$ and $b_B$. The interaction term corresponds to the between-group covariance between the direct and indirect mean channels.}
    \label{fig:schema:LM}
\end{figure}
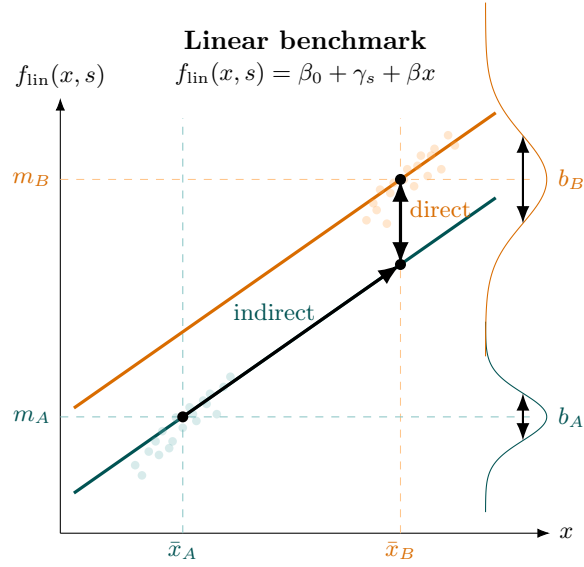

\subsection{Moment-aligning correction in the linear benchmark}
\label{subsec:lm-fair-correction}

The decomposition above is a diagnostic. It can be paired with a simple moment-aligning post-processing, which is useful for visualization and for quantifying the predictive cost of enforcing parity at the level of the first two score moments. Assume in this subsection that $b_s>0$ for all $s$. For each group $s$, define
\begin{equation}
\label{eq:linear-fair-correction}
  f^{\mathrm{fair}}_{\mathrm{lin}}(\bx,s)
  =
  \bar m+\frac{\bar b}{b_s}
  \{f_{\mathrm{lin}}(\bx,s)-m_s\}.
\end{equation}
Equivalently,
\[
  f^{\mathrm{fair}}_{\mathrm{lin}}(\bx,s)
  =
  \alpha_s f_{\mathrm{lin}}(\bx,s)+\delta_s,
  \qquad
  \alpha_s:=\frac{\bar b}{b_s},
  \qquad
  \delta_s:=\bar m-\alpha_s m_s.
\]

\begin{proposition}[Moment alignment]
\label{prop:linear-moment-alignment}
For every $s\in[M]$,
\[
  \E[f^{\mathrm{fair}}_{\mathrm{lin}}(\bX,S)\mid S=s]
  =
  \bar m,
  \qquad
  \sqrt{
  \Var(f^{\mathrm{fair}}_{\mathrm{lin}}(\bX,S)\mid S=s)}
  =
  \bar b.
\]
Consequently,
\[
  U_2(f^{\mathrm{fair}}_{\mathrm{lin}})=0.
\]
\end{proposition}

\begin{proof}
Conditionally on $S=s$,
\[
  \E[f^{\mathrm{fair}}_{\mathrm{lin}}(\bX,S)\mid S=s]
  =
  \bar m+\frac{\bar b}{b_s}(m_s-m_s)
  =
  \bar m,
\]
and
\[
  \Var(f^{\mathrm{fair}}_{\mathrm{lin}}(\bX,S)\mid S=s)
  =
  \left(\frac{\bar b}{b_s}\right)^2 b_s^2
  =
  \bar b^2.
\]
Both the conditional mean and the conditional standard deviation are therefore constant across groups, so $U_2(f^{\mathrm{fair}}_{\mathrm{lin}})=0$.
\end{proof}

On the coefficient scale, \eqref{eq:linear-fair-correction} gives, for group
$s$,
\[
  f^{\mathrm{fair}}_{\mathrm{lin}}(\bx,s)
  =
  \langle \bx,\alpha_s\bbeta\rangle
  +
  \alpha_s(\beta_0+\gamma_s)
  +
  \delta_s.
\]
Thus the non-sensitive slopes become group-specific:
\[
  \bbeta^{\mathrm{fair},s}=\alpha_s\bbeta,
\]
and the displayed shift for covariate $j$ is
\[
  \beta^{\mathrm{fair},s}_j-\beta_j
  =
  (\alpha_s-1)\beta_j.
\]
The group effect is absorbed into the group-specific intercept. The coefficient shift plot should therefore not be read as a new unconstrained refit, nor as a variable-importance plot. It shows how the original fitted score is contracted or expanded within each group in order to align the first two moments of the score distribution.

The same latent-scale correction will be used for GLMs in
Section~\ref{sec:glm-decomposition}. In that case the correction is applied to the latent score before the inverse-link transformation. Exact moment alignment then holds on the latent scale, while residual disparities may remain on the mean scale because of the nonlinear map.

\subsection{A Gaussian insurance illustration}
\label{subsec:lm-gaussian-illustration}

We illustrate the linear benchmark using the Medical Expenditure Panel Survey
Household Component (MEPS-HC), HC-192 file, corresponding to the 2016 Full
Year Consolidated Data File; details are reported in \ref{app:rep}.
The response is
\[
  y=\log(1+\texttt{TOTEXP16}),
\]
restricted here to individuals with positive annual health-care expenditures.
This restriction is specific to the identity-link benchmark, where the aim is to
obtain a simple Gaussian-style illustration. The GLM sections below consider
outcomes on their natural probability, count, or cost scales.

We consider two binary sensitive comparisons. The first is insurance coverage status, distinguishing insured from uninsured individuals. In a medical expenditure application, this variable is substantively relevant because it is linked to access to care and utilization, while also being correlated with socioeconomic vulnerability. The second is ethnicity, focusing on individuals identifying as Hispanic versus non-Hispanic individuals. These two comparisons allow us to distinguish a coverage-related disparity from an ethnicity-related disparity within the same empirical design.

Figure~\ref{fig:glm:identity:coeff} displays the coefficient-level effect of the
moment-aligning correction.

\begin{figure}[!htbp]
    \centering
    \includegraphics[width=\linewidth]{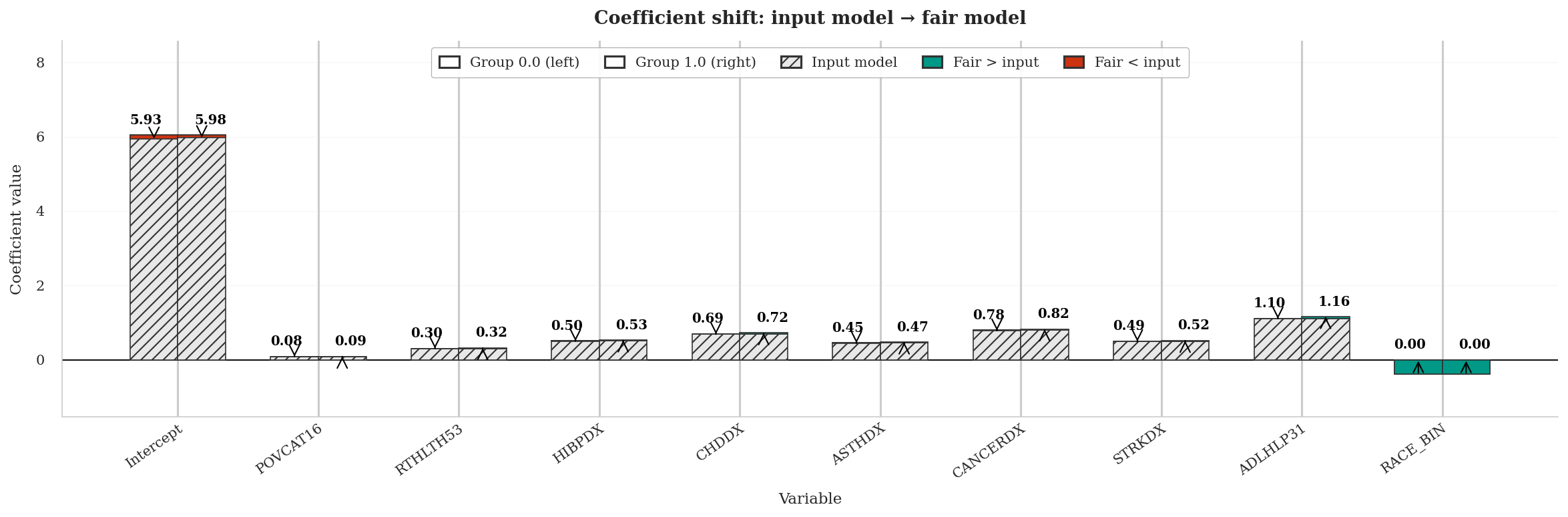}
    \includegraphics[width=\linewidth]{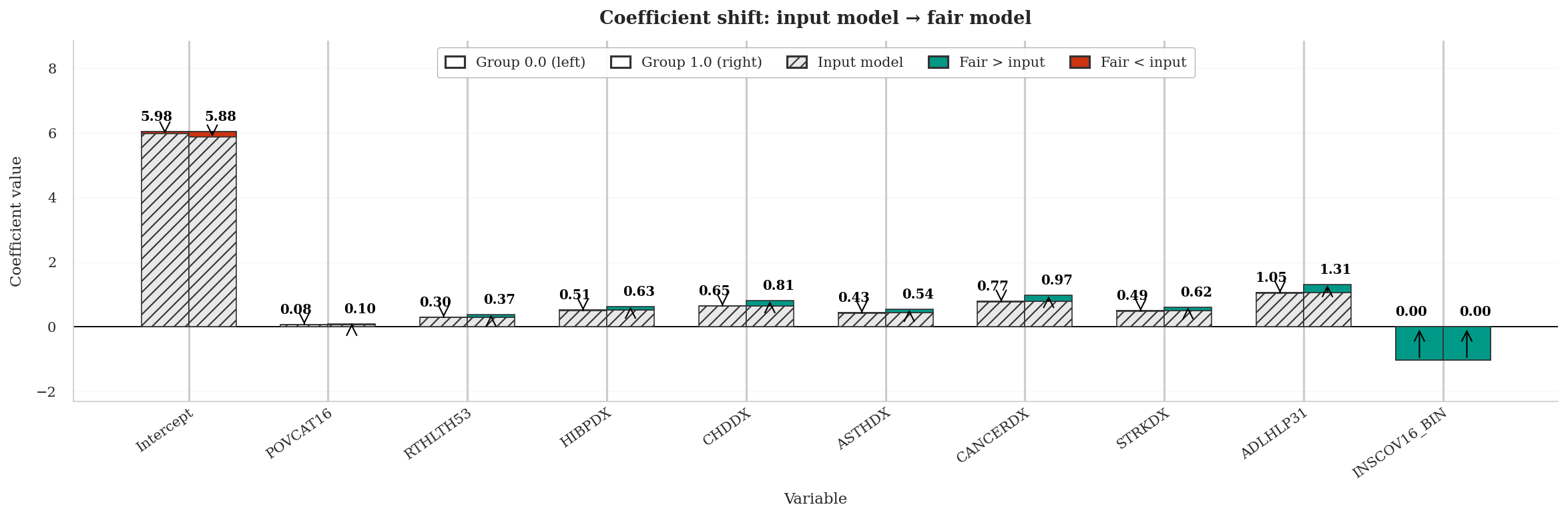}
    \caption{Linear model: coefficient shifts from the base model to the
    moment-aligned model. Hatched bars show input coefficients; colored
    extensions show the post-processing correction. The sensitive-attribute
    coefficient is absorbed into group-specific intercepts. The top panel uses
    the Hispanic/non-Hispanic comparison; the bottom panel uses the
    uninsured/insured comparison.}
    \label{fig:glm:identity:coeff}
\end{figure}

Table~\ref{tab:linear_decomposition} reports predictive performance, the value of $U_2$ before and after moment alignment, and the exact decomposition of the base-model disparity. The moment-aligning correction nearly eliminates $U_2$, as predicted by Proposition~\ref{prop:linear-moment-alignment}. The cost in predictive performance is moderate: the $R^2$ decreases and the RMSE increases, but the changes remain small relative to the baseline scale.

For both sensitive comparisons, the dominant channel is the direct mean component. The indirect mean component is small, which indicates that average differences in the retained non-sensitive covariates explain only a limited part of the prediction disparity. The interaction term is positive in both cases:
the direct group effect and the covariate-induced mean score differences reinforce each other. The structural component is negligible in this illustration, so most of the disparity is carried by mean shifts rather than by differences in within-group predictive dispersion.

\begin{table}[!htbp]
\centering
\footnotesize
\begin{tabular}{lcc}
\toprule
 & \textbf{Hispanic vs.} & \textbf{Uninsured} \\
& \textbf{non-Hispanic} & \textbf{vs.\ insured} \\
\midrule
\multicolumn{3}{l}{\textit{Predictive performance}} \\
\midrule
Base model: $R^2$
& 0.196 \,\scriptsize{(0.004)}
& 0.211 \,\scriptsize{(0.005)} \\
Moment-aligned model: $R^2$
& 0.178 \,\scriptsize{(0.005)}
& 0.172 \,\scriptsize{(0.005)} \\
Base model: RMSE
& 1.576 \,\scriptsize{(0.008)}
& 1.561 \,\scriptsize{(0.008)} \\
Moment-aligned model: RMSE
& 1.593 \,\scriptsize{(0.008)}
& 1.599 \,\scriptsize{(0.008)} \\
\midrule
\multicolumn{3}{l}{\textit{Moment-based disparity}} \\
\midrule
Base model: $U_2(f_{\mathrm{lin}})$
& 0.057 \,\scriptsize{(0.006)}
& 0.117 \,\scriptsize{(0.010)} \\
Moment-aligned model: $U_2(f_{\mathrm{lin}}^{\mathrm{fair}})$
& $< 0.001$
& $< 0.001$ \\
Relative reduction in $U_2$
& $\approx 100\%$
& $\approx 100\%$ \\
\midrule
\multicolumn{3}{l}{\textit{Decomposition of $U_2(f_{\mathrm{lin}})$}} \\
\midrule
Direct mean:
$\operatorname{Var}_S(\gamma_S)$
& 0.035 \,\scriptsize{(0.004)}
& 0.081 \,\scriptsize{(0.008)} \\
Indirect mean:
$\operatorname{Var}_S(\langle \bmu^{(S)},\bbeta\rangle)$
& 0.003 \,\scriptsize{(0.001)}
& 0.003 \,\scriptsize{(0.001)} \\
Interaction:
$2\operatorname{Cov}_S(\gamma_S,\langle \bmu^{(S)},\bbeta\rangle)$
& 0.019 \,\scriptsize{(0.002)}
& 0.032 \,\scriptsize{(0.003)} \\
Indirect structural:
$\operatorname{Var}_S(B)$
& 0.001 \,\scriptsize{(0.000)}
& 0.002 \,\scriptsize{(0.000)} \\
\midrule
Total decomposed disparity
& 0.057 \,\scriptsize{(0.006)}
& 0.117 \,\scriptsize{(0.010)} \\
\bottomrule
\end{tabular}
\caption{Linear regression on the logarithm of annual health-care expenditures,
$y=\log(1+\texttt{TOTEXP16})$, restricted to individuals with positive
expenditures. The table reports predictive performance of the base and
moment-aligned models, the moment-based disparity $U_2$, and its exact
decomposition into direct mean, indirect mean, interaction, and indirect
structural components. Entries are means over $30$ bootstrap replications, with
standard deviations in parentheses.}
\label{tab:linear_decomposition}
\end{table}

\subsection{Estimation and uncertainty}
\label{subsec:lm-estimation}

The decomposition is stated at the population level. In empirical applications, it is evaluated by plug-in. Let
\[
  \btheta:=(\bbeta,\bgamma,\beta_0)
\]
and consider a penalized least-squares estimator
\[
\widehat\btheta_\lambda
\in
\arg\min_{\btheta}
\left\{
\frac1n\sum_{i=1}^n
\bigl(y_i-\langle \bx_i,\bbeta\rangle-\gamma_{S_i}-\beta_0\bigr)^2
+
\lambda\mathcal P(\bbeta)
\right\},
\]
where $\mathcal P$ may be a ridge, lasso, or elastic-net penalty. The fitted predictor remains affine,
\[
\widehat f_{\lambda,\mathrm{lin}}(\bx,s)
=
\langle \bx,\widehat\bbeta_\lambda\rangle
+
\widehat\gamma_{\lambda,s}
+
\widehat\beta_{0,\lambda}.
\]
The empirical decomposition is obtained by replacing $\pi_s$, $\bmu^{(s)}$, $\bSigma^{(s)}$, $\bbeta$, and $\gamma_s$ by their empirical or fitted counterparts:
\[
\widehat U_2(\widehat f_{\lambda,\mathrm{lin}})
=
\Var_S(\widehat\gamma_{\lambda,S})
+
\Var_S\bigl(\langle \widehat\bmu^{(S)},\widehat\bbeta_\lambda\rangle\bigr)
+
2\,\Cov_S\bigl(\widehat\gamma_{\lambda,S},
\langle \widehat\bmu^{(S)},\widehat\bbeta_\lambda\rangle\bigr)
+
\Var_S(\widehat b_{\lambda,S}),
\]
with
\[
\widehat b_{\lambda,s}
=
\sqrt{
\widehat\bbeta_\lambda^\top
\widehat\bSigma^{(s)}
\widehat\bbeta_\lambda
}.
\]
Regularization changes the fitted coefficients and therefore the numerical size of the components, but not the algebraic form of the decomposition. Sampling uncertainty can be assessed by a nonparametric bootstrap: refit the model on each bootstrap sample, recompute the group-wise moments and decomposition terms, and report standard deviations, percentile intervals, or confidence bands
for the main components.

\section{Decomposition in Generalized Linear Models}
\label{sec:glm-decomposition}

We now move from the identity-link benchmark to generalized linear models.
The latent score remains additive in the covariates and in the sensitive attribute, but the actuarial prediction is the conditional mean obtained after applying an inverse link. Thus, the score-scale decomposition of Section~\ref{sec:lm-decomposition} is no longer exact on the prediction scale. This section derives the corresponding GLM diagnostic and makes explicit the two approximation layers involved.

The first layer maps the conditional moments of the latent score $\eta$ to
approximate conditional moments of $h(\eta)$ within each sensitive group. The
second layer expands the between-group dispersion of these proxy moments
around the global latent mean $\bar m$. We keep both errors visible. This
distinction is important in applications: the leading decomposition identifies
interpretable mechanisms, whereas the empirical criterion $U_2(f)$ remains the
quantity to report when approximation errors are non-negligible.

\subsection{Output-scale criterion and latent moments}
\label{subsec:glm-output-criterion}

We keep the notation of Section~\ref{sec:lm-decomposition}. Thus,
$\bX\in\mathbb R^k$ denotes the vector of non-sensitive covariates,
$S\in[M]$ the sensitive attribute, and
\[
  \bmu^{(s)}=\E[\bX\mid S=s],
  \qquad
  \bSigma^{(s)}=\Cov(\bX\mid S=s).
\]
Consider a GLM with conditional mean
\[
  \mu(\bx,s)=\E[Y\mid \bX=\bx,S=s],
\]
link function $g$, and inverse link $h=g^{-1}$. The prediction on the mean
scale is
\[
  f(\bx,s)=\mu(\bx,s)=h(\eta),
  \qquad
  \eta=\langle \bx,\bbeta\rangle+\gamma_s+\beta_0.
\]
The group-wise latent moments are
\begin{equation}
\label{eq:glm-latent-moments}
\begin{cases}
m_s
:=
\E[\eta\mid S=s]
=
\langle \bmu^{(s)},\bbeta\rangle+\gamma_s+\beta_0,
\\[0.3em]
v_s
:=
\Var(\eta\mid S=s)
=
\bbeta^\top\bSigma^{(s)}\bbeta.
\end{cases}
\end{equation}
We set
\[
  b_s:=\sqrt{v_s},
  \qquad
  M:=m_S,
  \qquad
  B:=b_S,
  \qquad
  \bar m:=\E_S[M],
  \qquad
  \bar b:=\E_S[B].
\]

Full distributional demographic parity would require equality of the conditional
prediction laws across sensitive groups. Outside the linear-Gaussian benchmark,
this condition is generally not reducible to finitely many moments. We therefore
use the following two-moment diagnostic on the prediction scale.

\begin{definition}[Output-scale moment disparity]
\label{def:glm-U2}
For any predictor $f$, define
\[
  \mu_f(s):=\E[f(\bX,S)\mid S=s],
  \qquad
  \sigma_f(s):=
  \sqrt{\Var(f(\bX,S)\mid S=s)}.
\]
The output-scale moment disparity of $f$ is
\begin{equation}
\label{eq:glm-U2}
  U_2(f)
  :=
  \Var_S\bigl(\mu_f(S)\bigr)
  +
  \Var_S\bigl(\sigma_f(S)\bigr).
\end{equation}
\end{definition}

The variance function of the GLM family does not enter
Definition~\ref{def:glm-U2}. The diagnostic is based on the distribution of the
fitted predictions $f(\bX,S)$, and its leading decomposition depends on the
inverse link $h$ and on the first two conditional moments of the latent score.
This is why Poisson and log-link Tweedie models have the same link-induced
transport, even though they correspond to different response distributions and
different actuarial quantities.

The latent mean path remains exactly decomposable:
\begin{equation}
\label{eq:glm-latent-mean-path}
\Var_S(M)
=
\underbrace{\Var_S(\gamma_S)}_{\textnormal{direct mean}}
+
\underbrace{\Var_S\bigl(\langle \bmu^{(S)},\bbeta\rangle\bigr)}
_{\textnormal{indirect mean}}
+
\underbrace{2\,\Cov_S\bigl(\gamma_S,
\langle \bmu^{(S)},\bbeta\rangle\bigr)}_{\textnormal{interaction}}.
\end{equation}
As in the linear case, the interaction term may be positive or negative. A
negative value means that the explicit group contribution and the covariate-driven
mean score differences partially offset each other. It is a statistical offset in
the fitted score, not a normative correction.

\subsection{Within-group transport to the mean scale}
\label{subsec:glm-within-group}

We first approximate the conditional moments of $h(\eta)$ within each group.

\begin{assumption}[Smooth inverse link and local remainder bounds]
\label{ass:glm-smoothness}
The inverse link $h$ is twice continuously differentiable on an open interval
$\mathcal D_h$, with $h'(t)>0$ on $\mathcal D_h$. For each $s\in[M]$,
$m_s\in\mathcal D_h$, and the conditional score $\eta\mid S=s$ takes values
in an interval $I_s\subset\mathcal D_h$ containing $m_s$. Writing
$\delta_s:=\eta-m_s$ conditionally on $S=s$,
assume
\[
  \rho_s:=\E[|\delta_s|^3\mid S=s]<\infty,
  \qquad
  \kappa_s:=\E[\delta_s^4\mid S=s]<\infty,
\]
and
\[
  L_{2,s}:=\sup_{t\in I_s}|h''(t)|<\infty.
\]
\end{assumption}

Assumption~\ref{ass:glm-smoothness} is stated as a sufficient condition for
explicit deterministic bounds. In empirical applications, the quantities
$U_2(f)$ and $\widetilde U_2(f)$ below are computed directly on the fitted
sample. The approximation should therefore be assessed empirically even when
the theoretical support of the score is unbounded, as with log-link models.

\begin{proposition}[Within-group delta approximation]
\label{prop:glm-delta}
Under Assumption~\ref{ass:glm-smoothness}, for each $s\in[M]$,
\begin{equation}
\label{eq:glm-mean-delta}
  \mu_f(s)
  =
  h(m_s)+e_\mu(s),
  \qquad
  |e_\mu(s)|\le \frac12 L_{2,s}v_s,
\end{equation}
and
\begin{equation}
\label{eq:glm-sd-delta}
  \sigma_f(s)
  =
  h'(m_s)b_s+e_\sigma(s),
  \qquad
  |e_\sigma(s)|
  \le
  \left(
  h'(m_s)L_{2,s}\rho_s
  +
  \frac14 L_{2,s}^2\kappa_s
  \right)^{1/2}.
\end{equation}
\end{proposition}

\begin{proof}
Fix $s$ and work conditionally on $S=s$. Write $\delta=\eta-m_s$. Taylor's
theorem gives
\[
  h(\eta)
  =
  h(m_s)+h'(m_s)\delta+\frac12h''(\xi)\delta^2,
\]
where $\xi$ lies between $m_s$ and $\eta$. Since
$\E[\delta\mid S=s]=0$,
\[
  \mu_f(s)-h(m_s)
  =
  \frac12\E[h''(\xi)\delta^2\mid S=s],
\]
which gives \eqref{eq:glm-mean-delta}.

For the standard deviation, write
\[
  h(\eta)
  =
  h(m_s)+h'(m_s)\delta+r,
  \qquad
  r:=\frac12h''(\xi)\delta^2.
\]
Then
\[
  \Var(h(\eta)\mid S=s)
  =
  \{h'(m_s)\}^2v_s
  +
  2h'(m_s)\Cov(\delta,r\mid S=s)
  +
  \Var(r\mid S=s).
\]
Since
\[
  |r|\le \frac12L_{2,s}\delta^2,
\]
we have
\[
  |\Cov(\delta,r\mid S=s)|
  \le
  \frac12L_{2,s}\rho_s,
  \qquad
  \Var(r\mid S=s)
  \le
  \frac14L_{2,s}^2\kappa_s.
\]
Thus
\[
\left|
  \Var(h(\eta)\mid S=s)
  -
  \{h'(m_s)\}^2v_s
\right|
\le
h'(m_s)L_{2,s}\rho_s
+
\frac14L_{2,s}^2\kappa_s.
\]
Using
\[
  |\sqrt u-\sqrt v|\le \sqrt{|u-v|},
  \qquad u,v\ge 0,
\]
gives \eqref{eq:glm-sd-delta}.
\end{proof}

This motivates the proxy moments
\begin{equation}
\label{eq:glm-proxy-moments}
  \widetilde\mu_f(s):=h(m_s),
  \qquad
  \widetilde\sigma_f(s):=h'(m_s)b_s,
\end{equation}
and the associated proxy criterion
\begin{equation}
\label{eq:glm-U2-tilde}
  \widetilde U_2(f)
  :=
  \Var_S\bigl(\widetilde\mu_f(S)\bigr)
  +
  \Var_S\bigl(\widetilde\sigma_f(S)\bigr).
\end{equation}
We also write
\[
  \widetilde U_2(f)
  =
  \widetilde{\FMD}(f)
  +
  \widetilde{\SMD}(f),
\]
where
\[
  \widetilde{\FMD}(f):=\Var_S\bigl(h(m_S)\bigr),
  \qquad
  \widetilde{\SMD}(f):=\Var_S\bigl(h'(m_S)b_S\bigr).
\]

The first approximation layer is the difference between $U_2(f)$ and
$\widetilde U_2(f)$.

\begin{proposition}[Within-group approximation error]
\label{prop:glm-within-error}
Let
\[
  e_\mu(s):=\mu_f(s)-\widetilde\mu_f(s),
  \qquad
  e_\sigma(s):=\sigma_f(s)-\widetilde\sigma_f(s).
\]
Then
\begin{align}
  U_2(f)
  &=
  \widetilde U_2(f)
  +
  \Delta_{\mathrm{within}}(f),
  \label{eq:glm-within-error}
\end{align}
where
\begin{align}
  \Delta_{\mathrm{within}}(f)
  &:=
  2\,\Cov_S\bigl(\widetilde\mu_f(S),e_\mu(S)\bigr)
  +
  \Var_S(e_\mu(S))
  \notag\\
  &\quad+
  2\,\Cov_S\bigl(\widetilde\sigma_f(S),e_\sigma(S)\bigr)
  +
  \Var_S(e_\sigma(S)).
  \label{eq:glm-delta-within-definition}
\end{align}
\end{proposition}

\begin{proof}
Use
\[
  \mu_f(S)=\widetilde\mu_f(S)+e_\mu(S),
  \qquad
  \sigma_f(S)=\widetilde\sigma_f(S)+e_\sigma(S),
\]
and expand the two between-group variances in Definition~\ref{def:glm-U2}.
\end{proof}

\subsection{Across-group expansion and leading decomposition}
\label{subsec:glm-across-groups}

The proxy criterion $\widetilde U_2(f)$ is still nonlinear in the latent group
means $m_s$. We now expand it around $\bar m$.

\begin{assumption}[Across-group smoothness]
\label{ass:glm-across-smoothness}
The inverse link satisfies $h\in\mathcal C^3(J)$ on an open interval $J$
containing $\bar m$ and all group means $m_s$, $s\in[M]$.
\end{assumption}

For each group, define
\[
  \Delta_s:=m_s-\bar m,
  \qquad
  \Gamma_s:=b_s-\bar b.
\]
The next proposition isolates the leading terms and the across-group remainder.

\begin{proposition}[Across-group expansion of the proxy criterion]
\label{prop:glm-across}
Under Assumption~\ref{ass:glm-across-smoothness},
\begin{equation}
\label{eq:glm-FMD-expansion}
  \widetilde{\FMD}(f)
  =
  \{h'(\bar m)\}^2\Var_S(M)
  +
  R_{\FMD},
\end{equation}
and
\begin{align}
\label{eq:glm-SMD-expansion}
  \widetilde{\SMD}(f)
  &=
  \{h'(\bar m)\}^2\Var_S(B)
  +
  2h'(\bar m)h''(\bar m)\bar b\,\Cov_S(M,B)
  \notag\\
  &\quad+
  \{h''(\bar m)\}^2\bar b^2\Var_S(M)
  +
  R_{\SMD}.
\end{align}
\end{proposition}

\begin{proof}
A second-order expansion gives
\[
  h(m_s)
  =
  h(\bar m)+h'(\bar m)\Delta_s+r_s,
  \qquad
  r_s=\frac12h''(\xi_s)\Delta_s^2,
\]
where $\xi_s$ lies between $m_s$ and $\bar m$. Since constants vanish under
$\Var_S$,
\[
  \widetilde{\FMD}(f)
  =
  \Var_S\bigl(h'(\bar m)\Delta_S+r_S\bigr),
\]
and therefore
\[
  R_{\FMD}
  =
  2h'(\bar m)\Cov_S(\Delta_S,r_S)
  +
  \Var_S(r_S).
\]

For the second term, expand
\[
  h'(m_s)
  =
  h'(\bar m)+h''(\bar m)\Delta_s+\varepsilon_s,
  \qquad
  \varepsilon_s=\frac12h'''(\zeta_s)\Delta_s^2,
\]
where $\zeta_s$ lies between $m_s$ and $\bar m$. Then
\[
  h'(m_s)b_s
  =
  h'(\bar m)\bar b
  +
  L_s
  +
  E_s,
\]
with
\[
  L_s
  :=
  h'(\bar m)\Gamma_s
  +
  h''(\bar m)\bar b\,\Delta_s,
\]
and
\[
  E_s
  :=
  h''(\bar m)\Delta_s\Gamma_s
  +
  \varepsilon_s\bar b
  +
  \varepsilon_s\Gamma_s.
\]
Hence
\[
  \widetilde{\SMD}(f)
  =
  \Var_S(L_S)
  +
  2\Cov_S(L_S,E_S)
  +
  \Var_S(E_S).
\]
Expanding $\Var_S(L_S)$ gives the three leading terms in
\eqref{eq:glm-SMD-expansion}, and
\[
  R_{\SMD}
  =
  2\Cov_S(L_S,E_S)
  +
  \Var_S(E_S).
\]
\end{proof}

Define the linear latent disparity
\begin{equation}
\label{eq:glm-latent-linear-disparity}
  D_{\mathrm{lin}}
  :=
  \Var_S(\gamma_S)
  +
  \Var_S\bigl(\langle \bmu^{(S)},\bbeta\rangle\bigr)
  +
  2\,\Cov_S\bigl(\gamma_S,\langle \bmu^{(S)},\bbeta\rangle\bigr)
  +
  \Var_S(B).
\end{equation}
This is exactly the four-way linear quantity from
Theorem~\ref{thm:linear-four-way}, evaluated on the latent GLM score.

\begin{definition}[Leading GLM decomposition]
\label{def:glm-leading-decomposition}
The leading GLM decomposition is
\begin{align}
  D_1(f)
  &:=
  \{h'(\bar m)\}^2
  \Big[
  \underbrace{\Var_S(\gamma_S)}_{\textnormal{direct mean}}
  +
  \underbrace{\Var_S\bigl(\langle \bmu^{(S)},\bbeta\rangle\bigr)}
  _{\textnormal{indirect mean}}
  \notag\\
  &\qquad\qquad+
  \underbrace{2\,\Cov_S\bigl(\gamma_S,
  \langle \bmu^{(S)},\bbeta\rangle\bigr)}_{\textnormal{interaction}}
  +
  \underbrace{\Var_S(B)}_{\textnormal{indirect structural}}
  \Big]
  \notag\\
  &\quad+
  \underbrace{
  2h'(\bar m)h''(\bar m)\bar b\,\Cov_S(M,B)}
  _{\textnormal{curvature coupling}}
  +
  \underbrace{
  \{h''(\bar m)\}^2\bar b^2\Var_S(M)}
  _{\textnormal{curvature amplification}}.
  \label{eq:glm-D1}
\end{align}
\end{definition}

The bracketed term is the linear latent decomposition, rescaled by the local
slope of the inverse link. The last two terms are link-induced. The curvature
coupling term records whether group differences in average latent scores and
group differences in latent score dispersion move together after the nonlinear
transport. The curvature amplification term is nonnegative and measures the
additional output-scale dispersion created by curvature when latent group means
differ.

\begin{theorem}[Two-layer GLM decomposition]
\label{thm:glm-two-layer}
Under Assumptions~\ref{ass:glm-smoothness}
and~\ref{ass:glm-across-smoothness},
\begin{equation}
\label{eq:glm-two-layer}
  U_2(f)
  =
  D_1(f)
  +
  \Delta_{\mathrm{within}}(f)
  +
  \Delta_{\mathrm{across}}(f),
\end{equation}
where
\[
  \Delta_{\mathrm{within}}(f)
  =
  U_2(f)-\widetilde U_2(f),
\]
with the exact expression given in
\eqref{eq:glm-delta-within-definition}, and
\[
  \Delta_{\mathrm{across}}(f)
  =
  R_{\FMD}+R_{\SMD}
  =
  \widetilde U_2(f)-D_1(f).
\]
\end{theorem}

\begin{proof}
By Proposition~\ref{prop:glm-within-error},
\[
  U_2(f)
  =
  \widetilde U_2(f)
  +
  \Delta_{\mathrm{within}}(f).
\]
By Proposition~\ref{prop:glm-across} and
\eqref{eq:glm-latent-mean-path},
\[
  \widetilde U_2(f)
  =
  D_1(f)
  +
  R_{\FMD}
  +
  R_{\SMD}.
\]
Combining the two identities gives \eqref{eq:glm-two-layer}.
\end{proof}

\begin{corollary}[Identity-link benchmark]
\label{cor:glm-identity}
If $h(t)=t$, then
\[
  U_2(f)=\widetilde U_2(f)=D_1(f),
\]
all curvature and remainder terms vanish, and
Theorem~\ref{thm:glm-two-layer} reduces to the exact four-way linear
decomposition of Theorem~\ref{thm:linear-four-way}.
\end{corollary}

\begin{proof}
For $h(t)=t$, one has $h'(t)=1$ and $h''(t)=0$. Hence
\[
  \widetilde\mu_f(s)=m_s,
  \qquad
  \widetilde\sigma_f(s)=b_s,
\]
and the two approximation errors vanish.
\end{proof}

\begin{table}[!htbp]
\centering
\caption{Diagnostic interpretation of the leading terms in
Definition~\ref{def:glm-leading-decomposition}.}
\label{tab:glm-diagnostic-terms}
\fontsize{9}{11}\selectfont
\begin{tabular}{p{3.6cm}p{9.2cm}}
\toprule
\textbf{Term} & \textbf{Diagnostic reading} \\
\midrule
Direct mean
& Explicit fitted contribution of the sensitive attribute through $\gamma_S$. \\
Indirect mean
& Group differences in average non-sensitive covariate profiles, projected
through $\bbeta$. \\
Interaction
& Reinforcement or offset between the direct and indirect mean channels. \\
Indirect structural
& Group differences in within-group latent score dispersion, induced by
$\bSigma^{(S)}$. \\
Curvature coupling
& Interaction between group-wise latent score means and latent score
standard deviations through $h''(\bar m)$. \\
Curvature amplification
& Additional output-scale disparity generated by curvature and
between-group latent mean heterogeneity. \\
\bottomrule
\end{tabular}
\end{table}

Figure~\ref{fig:schema:GLM} summarizes the transport from the latent scale to
the mean scale. The tangent used in the first-order mapping is
\[
  T_{\bar m}(t)
  =
  h(\bar m)+h'(\bar m)(t-\bar m).
\]
The leading decomposition is reliable when this local linearization captures the
relevant score range. When the score is highly dispersed, when group means are
far from $\bar m$, or when $h$ has strong curvature on the fitted range, the
terms $\Delta_{\mathrm{within}}(f)$ and $\Delta_{\mathrm{across}}(f)$ should be
reported and interpreted.

\begin{figure}[!htbp]
    \centering
    \begin{tikzpicture}[
  >=Latex,
  scale=.9,
  every node/.style={font=\small},
  tealA/.style={color=teal!65!black},
  orangeB/.style={color=orange!85!black}
]

\def\betap{0.45}
\def\betaZero{0.20}
\def\gammaA{0.00}
\def\gammaB{1}

\def\xA{1.5}
\def\xB{3.8}

\pgfmathsetmacro{\mA}{\betaZero+\gammaA+\betap*\xA}
\pgfmathsetmacro{\mB}{\betaZero+\gammaB+\betap*\xB}
\pgfmathsetmacro{\mAatB}{\betaZero+\gammaA+\betap*\xB}

\def\bA{0.18}
\def\bB{0.28}

\pgfmathsetmacro{\muA}{exp(\mA)}
\pgfmathsetmacro{\muB}{exp(\mB)}
\pgfmathsetmacro{\muAatB}{exp(\mAatB)}

\pgfmathsetmacro{\mbar}{0.5*(\mA+\mB)}
\pgfmathsetmacro{\mubar}{exp(\mbar)}

\def\xdL{5.45}
\def\denswidthL{0.36}

\def\shiftR{8.2}


\draw[->] (0,0) -- (6.3,0) node[right] {$x$};
\draw[->] (0,0) -- (0,4.15) node[above] {$\eta$};

\draw[teal!45, dashed] (\xA,0) -- (\xA,3.95);
\draw[orange!45, dashed] (\xB,0) -- (\xB,3.95);

\draw[teal!45, dashed] (0,\mA) -- (6.0,\mA);
\draw[orange!45, dashed] (0,\mB) -- (6.0,\mB);

\node[teal!65!black, below] at (\xA,0) {$\bar x_A$};
\node[orange!85!black, below] at (\xB,0) {$\bar x_B$};

\node[teal!65!black, left] at (0,\mA) {$m_A$};
\node[orange!85!black, left] at (0,\mB) {$m_B$};

\draw[teal!65!black, very thick]
  plot[domain=0.15:5.2, samples=2]
  (\x,{\betaZero+\gammaA+\betap*\x});

\draw[orange!85!black, very thick]
  plot[domain=0.15:5.2, samples=2]
  (\x,{\betaZero+\gammaB+\betap*\x});

\def\hya{-0.3}
\def\hyb{0.10}

\foreach \x/\y in {
0.9/0.75,1.0/0.90,1.1/1.05,1.2/0.95,1.3/1.15,
1.4/1.00,1.5/1.18,1.6/1.30,1.7/1.22,1.8/1.42,
1.2/1.20,1.35/0.85,1.55/1.10,1.75/1.32,1.95/1.45
}{
  \fill[teal!55, opacity=.28] (\x,{\y+\hya}) circle (1.8pt);
}

\foreach \x/\y in {
3.3/2.25,3.4/2.50,3.5/2.65,3.6/2.75,3.7/2.55,
3.8/2.85,3.9/2.95,4.0/2.70,4.1/3.05,4.2/2.90,
3.45/2.85,3.65/2.40,3.85/2.78,4.05/3.00,4.20/3.15
}{
  \fill[orange!65, opacity=.30] (\x,{\y+\hyb}) circle (1.8pt);
}

\fill[black] (\xA,\mA) circle (2.3pt);
\fill[black] (\xB,\mAatB) circle (2.3pt);
\fill[black] (\xB,\mB) circle (2.3pt);

\draw[black, very thick, ->]
  (\xA,\mA) -- (\xB,\mAatB);
\node[teal!65!black] at (2.75,2.55) {indirect};

\draw[black, very thick, <->]
  (\xB,\mAatB) -- (\xB,\mB);
\node[orange!85!black, right] at (\xB,{(\mAatB+\mB)/2+0.12}) {direct};

\draw[teal!65!black, thin]
  plot[domain={\mA-4*\bA}:{\mA+4*\bA}, samples=100, variable=\y]
  ({\xdL + 2*\denswidthL*exp(-((\y-\mA)^2)/(2*\bA*\bA))}, {\y});

\draw[orange!85!black, thin]
  plot[domain={\mB-4*\bB}:{\mB+4*\bB}, samples=100, variable=\y]
  ({\xdL + 2*\denswidthL*exp(-((\y-\mB)^2)/(2*\bB*\bB))}, {\y});

\draw[black, thick, <->]
  ({\xdL+\denswidthL+.08},{\mA-2*\bA}) --
  ({\xdL+\denswidthL+.08},{\mA+2*\bA});
\node[teal!65!black, right] at ({\xdL+\denswidthL+.35},\mA) {$b_A$};

\draw[black, thick, <->]
  ({\xdL+\denswidthL+.08},{\mB-2*\bB}) --
  ({\xdL+\denswidthL+.08},{\mB+2*\bB});
\node[orange!85!black, right] at ({\xdL+\denswidthL+.35},\mB) {$b_B$};

\node[font=\bfseries] at (3.1,5.25) {Latent scale};
\node[font=\small] at (3.1,4.73)
  {$\eta(x,s)=\beta_0+\gamma_s+\beta x$};

\def\yfacR{0.28}

\def\gammaB{0.50}
\def\xA{1.5}
\def\xB{3.8}

\pgfmathsetmacro{\mA}{\betaZero+\gammaA+\betap*\xA}
\pgfmathsetmacro{\mB}{\betaZero+\gammaB+\betap*\xB}
\pgfmathsetmacro{\mAatB}{\betaZero+\gammaA+\betap*\xB}

\def\bA{0.18}
\def\bB{0.28}

\pgfmathsetmacro{\muA}{exp(\mA)}
\pgfmathsetmacro{\muB}{exp(\mB)}
\pgfmathsetmacro{\muAatB}{exp(\mAatB)}

\pgfmathsetmacro{\mbar}{0.5*(\mA+\mB)}
\pgfmathsetmacro{\mubar}{exp(\mbar)}

\pgfmathsetmacro{\muAplot}{\yfacR*\muA}
\pgfmathsetmacro{\muBplot}{\yfacR*\muB}
\pgfmathsetmacro{\mubarplot}{\yfacR*\mubar}

\pgfmathsetmacro{\sigAplot}{\yfacR*\muA*\bA}
\pgfmathsetmacro{\sigBplot}{\yfacR*\muB*\bB}

\begin{scope}[shift={(\shiftR,0)}]

\draw[->] (0,0) -- (4.1,0) node[right] {$\eta$};
\draw[->] (0,0) -- (0,5.3) node[above] {$\mu$};

\draw[black, very thick]
  plot[domain=0.15:2.85, samples=100]
  (\x,{\yfacR*exp(\x)});

\draw[black!60, dashed, thick]
  plot[domain=0.15:3.0, samples=2]
  (\x,{\yfacR*(\mubar + \mubar*(\x-\mbar))});

\draw[teal!45, dashed] (\mA,0) -- (\mA,4.95);
\draw[orange!45, dashed] (\mB,0) -- (\mB,4.95);
\draw[black!35, dotted] (\mbar,0) -- (\mbar,4.95);

\draw[teal!45, dashed] (0,\muAplot) -- (3.9,\muAplot);
\draw[orange!45, dashed] (0,\muBplot) -- (3.9,\muBplot);

\fill[teal!65!black] (\mA,\muAplot) circle (2.6pt);
\fill[orange!85!black] (\mB,\muBplot) circle (2.6pt);
\fill[black] (\mbar,\mubarplot) circle (2.4pt);

\node[teal!65!black, below] at (\mA,0) {$m_A$};
\node[orange!85!black, below] at (\mB,0) {$m_B$};

\node[teal!65!black, left] at (0,\muAplot) {$\tilde\mu_A$};
\node[orange!85!black, left] at (0,\muBplot) {$\tilde\mu_B$};

\draw[teal!65!black, thick, <->]
  ({\mA+0.10},{\muAplot-3*\sigAplot}) --
  ({\mA+0.10},{\muAplot+3*\sigAplot});
\node[teal!65!black, right] at ({\mA+0.2},{\muAplot-\sigAplot-.2})
  {$\tilde\sigma_A$};

\draw[orange!85!black, thick, <->]
  ({\mB+0.10},{\muBplot-\sigBplot}) --
  ({\mB+0.10},{\muBplot+\sigBplot});
\node[orange!85!black, right] at ({\mB+0.2},{\muBplot-\sigBplot})
  {$\tilde\sigma_B$};

\node[black!70] at (3.50,3.35) {tangent};



\node[font=\bfseries] at (2.05,5.75) {Mean scale};
\node[font=\small] at (2.05,5.25)
  {$h(\eta)=e^\eta$};

\end{scope}

\node[font=\bfseries] at (6.5,6.5)
  {GLM extension: latent to mean-scale transport};

\end{tikzpicture}
    \caption{Schematic extension to generalized linear models. On the latent
    scale, group differences are read through direct and indirect mean channels,
    while $b_A$ and $b_B$ summarize within-group score dispersion. The inverse
    link $h$ transports these latent quantities to the mean scale. The tangent
    $T_{\bar m}(t)=h(\bar m)+h'(\bar m)(t-\bar m)$ illustrates the local
    first-order mapping; curvature generates the additional coupling and
    amplification terms.}
    \label{fig:schema:GLM}
\end{figure}
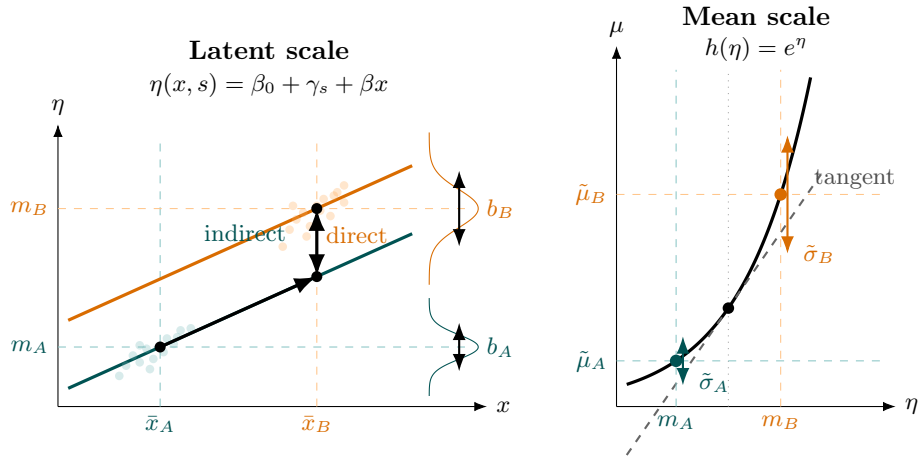

\subsection{Latent-scale moment alignment for GLMs}
\label{subsec:glm-fair-correction}

The decomposition above is a diagnostic. As in the linear benchmark, it can be
paired with a moment-aligning post-processing. In a GLM, this correction must
be applied to the latent score, not directly to the prediction $h(\eta)$.

Assume that $b_s>0$ for all $s$. For each group $s$, define
\begin{equation}
\label{eq:glm-fair-eta}
  \eta^{\mathrm{fair}}
  =
  \bar m+\frac{\bar b}{b_s}\{\eta-m_s\}
  =
  \alpha_s\eta+\delta_s,
  \qquad
  \alpha_s:=\frac{\bar b}{b_s},
  \qquad
  \delta_s:=\bar m-\alpha_s m_s.
\end{equation}
The corrected GLM prediction is then
\begin{equation}
\label{eq:glm-fair-prediction}
  f^{\mathrm{fair}}(\bx,s)
  =
  h\{\eta^{\mathrm{fair}}(\bx,s)\}.
\end{equation}

\begin{proposition}[Latent moment alignment]
\label{prop:glm-latent-moment-alignment}
For every $s\in[M]$,
\[
  \E[\eta^{\mathrm{fair}}\mid S=s]=\bar m,
  \qquad
  \sqrt{\Var(\eta^{\mathrm{fair}}\mid S=s)}=\bar b.
\]
Consequently, the first-order proxy disparity of the corrected predictor is
zero:
\[
  \widetilde U_2(f^{\mathrm{fair}})=0.
\]
\end{proposition}

\begin{proof}
The first two identities follow directly from
\eqref{eq:glm-fair-eta}:
\[
  \E[\eta^{\mathrm{fair}}\mid S=s]
  =
  \bar m+\frac{\bar b}{b_s}(m_s-m_s)
  =
  \bar m,
\]
and
\[
  \Var(\eta^{\mathrm{fair}}\mid S=s)
  =
  \left(\frac{\bar b}{b_s}\right)^2b_s^2
  =
  \bar b^2.
\]
Therefore, for all $s$,
\[
  \widetilde\mu_{f^{\mathrm{fair}}}(s)=h(\bar m),
  \qquad
  \widetilde\sigma_{f^{\mathrm{fair}}}(s)=h'(\bar m)\bar b.
\]
Both proxy moments are constant across groups, so
$\widetilde U_2(f^{\mathrm{fair}})=0$.
\end{proof}

The equality $\widetilde U_2(f^{\mathrm{fair}})=0$ is a first-order statement.
The exact output-scale disparity $U_2(f^{\mathrm{fair}})$ may remain positive
when the conditional latent score distributions differ beyond their first two
moments. If the affine correction equalizes the full conditional latent score
distributions, for instance under Gaussian location-scale score distributions,
then the conditional distributions of $h(\eta^{\mathrm{fair}})$ are also equal
and $U_2(f^{\mathrm{fair}})=0$.

On the coefficient scale, \eqref{eq:glm-fair-eta} gives, for group $s$,
\[
  \eta^{\mathrm{fair}}(\bx,s)
  =
  \langle \bx,\alpha_s\bbeta\rangle
  +
  \alpha_s(\beta_0+\gamma_s)
  +
  \delta_s.
\]
Thus the non-sensitive slopes become group-specific,
\[
  \bbeta^{\mathrm{fair},s}=\alpha_s\bbeta,
\]
and the group-specific intercept absorbs the original sensitive effect and the
recentering term. Coefficient-shift figures should therefore be read as latent
score transformations. Their effect on probabilities, frequencies, or premiums
depends on the slope and curvature of $h$ at the relevant operating point.

\subsection{Assessing the approximation in practice}
\label{subsec:glm-approximation-assessment}

The approximation can be assessed empirically because $U_2(f)$,
$\widetilde U_2(f)$, and $D_1(f)$ are all computable after fitting the model.
We recommend reporting
  $U_2(f)$,
  $\widetilde U_2(f)$,
  $D_1(f)$,
together with the signed approximation errors
\[
  \Delta_{\mathrm{within}}(f)
  :=
  U_2(f)-\widetilde U_2(f),
  \qquad
  \Delta_{\mathrm{across}}(f)
  :=
  \widetilde U_2(f)-D_1(f),
\]
and, when $U_2(f)>0$, the unsigned relative discrepancy
\begin{equation}
\label{eq:r-within}
  r_{\mathrm{within}}(f)
  :=
  \frac{|U_2(f)-\widetilde U_2(f)|}{U_2(f)}.
\end{equation}
The signed error $\Delta_{\mathrm{within}}$ may be positive or negative. The
relative discrepancy $r_{\mathrm{within}}$ is nonnegative by definition.

When $r_{\mathrm{within}}$ is small and $\Delta_{\mathrm{across}}$ is negligible,
the leading components in $D_1(f)$ can be read as a faithful decomposition of
the observed output-scale moment disparity. When these discrepancies are
moderate or large, the components remain useful as directional diagnostics, but
the exact empirical criterion $U_2(f)$ should be reported and interpreted
alongside the leading decomposition.

Large discrepancies are expected when the latent score is highly dispersed
within groups, when group means $m_s$ are far from $\bar m$, or when the
inverse link has strong curvature on the fitted score range. These are precisely
the situations reflected in the bounds of Proposition~\ref{prop:glm-delta} and
in the remainders of Proposition~\ref{prop:glm-across}. The approximation is
therefore not a hidden modelling assumption; it is a diagnostic layer whose
accuracy can be checked.

\subsection{Extension to generalized additive models}
\label{subsec:glm-gam-extension}

The same construction applies when the latent score is additive rather than
linear. Consider a GAM predictor
\begin{equation}
\label{eq:gam-eta}
  \eta
  =
  \sum_{j=1}^k f_j(x_j)+\gamma_s+\beta_0,
  \qquad
  \mu(\bx,s)=h(\eta).
\end{equation}
For each component $j$, define
\[
  m_{j,s}:=\E[f_j(X_j)\mid S=s],
  \qquad
  v_{j,s}:=\Var(f_j(X_j)\mid S=s),
\]
and, for $i\neq j$,
\[
  c_{ij,s}:=\Cov(f_i(X_i),f_j(X_j)\mid S=s).
\]
The component functions are understood under the usual centering or
identifiability constraints for smooth terms
\cite{hastie1990generalized,wood2017generalized}.

\begin{proposition}[Latent moments in a GAM]
\label{prop:gam-component-decomp}
For the latent score \eqref{eq:gam-eta},
\begin{equation}
\label{eq:gam-ms}
  m_s
  =
  \sum_{j=1}^k m_{j,s}
  +
  \gamma_s
  +
  \beta_0,
\end{equation}
and
\begin{equation}
\label{eq:gam-vs}
  v_s
  =
  \sum_{j=1}^k v_{j,s}
  +
  2\sum_{1\le i<j\le k} c_{ij,s}.
\end{equation}
Consequently, if $M_j:=m_{j,S}$, then
\begin{align}
\label{eq:gam-VarM}
  \Var_S(M)
  &=
  \Var_S(\gamma_S)
  +
  \sum_{j=1}^k \Var_S(M_j)
  +
  2\sum_{j=1}^k \Cov_S(\gamma_S,M_j)
  \notag\\
  &\quad+
  2\sum_{1\le i<j\le k}\Cov_S(M_i,M_j).
\end{align}
\end{proposition}

\begin{proof}
Equations \eqref{eq:gam-ms} and \eqref{eq:gam-vs} follow by taking conditional
expectations and variances in \eqref{eq:gam-eta}. Equation
\eqref{eq:gam-VarM} follows by expanding the between-group variance of
\[
  M=\sum_{j=1}^k M_j+\gamma_S+\beta_0.
\]
\end{proof}

Thus, Theorem~\ref{thm:glm-two-layer} applies to GAMs after replacing the
linear quantities $m_s$ and $b_s$ by the additive expressions
\eqref{eq:gam-ms} and \eqref{eq:gam-vs}. A smooth component contributes
through its between-group mean differences, its within-group variance, and its
covariance with the other components.

\subsection{Estimation and empirical reporting}
\label{subsec:glm-estimation}

The decomposition is evaluated by plug-in after fitting the model. For a fitted
GLM or GAM, compute the empirical analogues of $m_s$, $b_s$, $\bar m$, and
$\bar b$, and substitute them into $D_1(f)$ in
Definition~\ref{def:glm-leading-decomposition}. Penalization changes the
estimated coefficients or smooth components, but not the algebraic form of the
diagnostic. The same applies to ridge, lasso, elastic-net, and penalized spline
estimation.

Sampling uncertainty can be assessed by a nonparametric bootstrap: refit the
model on each bootstrap sample, recompute $U_2(f)$, $\widetilde U_2(f)$,
$D_1(f)$, and all leading components, and report standard deviations or
percentile intervals. This is especially useful when some sensitive groups are
small, since estimates of $\bSigma^{(s)}$ and of the structural component may
then be unstable.

In the numerical sections below, the tables should distinguish three quantities:
the exact empirical output-scale criterion $U_2(f)$, the proxy
$\widetilde U_2(f)$, and the leading decomposition $D_1(f)$. The signed gap
$U_2(f)-\widetilde U_2(f)$ may be positive or negative, whereas the reported
relative discrepancy must be nonnegative, as in \eqref{eq:r-within}.

\section{The Logistic Case}
\label{sec:logistic}

Logistic regression is the canonical GLM for binary actuarial outcomes, such as
claim occurrence, hospitalization, default, lapse, or renewal. It is also the
simplest nonlinear case after the identity link. The prediction is bounded
between zero and one, and the inverse link has a changing curvature. The
logistic case therefore illustrates how a latent-score disparity can be compressed
or reshaped when it is read on a probability scale.

\subsection{Logistic GLM and inverse-logit derivatives}
\label{subsec:logistic-setup}

Let
\[
  Y\mid(\bX,S)\sim \mathrm{Bernoulli}\bigl(\mu(\bX,S)\bigr),
  \qquad
  \E[Y\mid\bX,S]=\mu(\bX,S).
\]
With the canonical logit link,
\[
  g(\mu)=\log\left(\frac{\mu}{1-\mu}\right),
\]
the latent score and the predicted probability are
\[
  \eta
  =
  \langle\bx,\bbeta\rangle+\gamma_s+\beta_0,
  \qquad
  f(\bx,s)=\mu(\bx,s)=h(\eta),
\]
where
\begin{equation}
\label{eq:logistic-inverse-link}
  h(t)=\frac{1}{1+e^{-t}}.
\end{equation}
As in Section~\ref{sec:glm-decomposition},
\[
  m_s
  =
  \E[\eta\mid S=s]
  =
  \langle\bmu^{(s)},\bbeta\rangle+\gamma_s+\beta_0,
  \qquad
  b_s
  =
  \sqrt{\bbeta^\top\bSigma^{(s)}\bbeta},
\]
and
\[
  M:=m_S,
  \qquad
  B:=b_S,
  \qquad
  \bar m:=\E_S[M],
  \qquad
  \bar b:=\E_S[B].
\]

\begin{lemma}[Derivatives of the inverse-logit map]
\label{lem:logistic-derivatives}
For $h(t)=(1+e^{-t})^{-1}$,
\[
  h'(t)=h(t)\bigl(1-h(t)\bigr),
  \qquad
  h''(t)=h(t)\bigl(1-h(t)\bigr)\bigl(1-2h(t)\bigr).
\]
In particular, if
\[
  \bar\mu:=h(\bar m),
  \qquad
  a:=\bar\mu(1-\bar\mu),
\]
then
\begin{equation}
\label{eq:logistic-derivatives-at-mbar}
  h'(\bar m)=a,
  \qquad
  h''(\bar m)=a(1-2\bar\mu).
\end{equation}
\end{lemma}

\begin{proof}
Differentiating $h(t)=(1+e^{-t})^{-1}$ gives
\[
  h'(t)
  =
  \frac{e^{-t}}{(1+e^{-t})^2}
  =
  h(t)\bigl(1-h(t)\bigr).
\]
Differentiating once more yields
\[
  h''(t)
  =
  h'(t)\bigl(1-2h(t)\bigr)
  =
  h(t)\bigl(1-h(t)\bigr)\bigl(1-2h(t)\bigr).
\]
Evaluation at $t=\bar m$ gives
\eqref{eq:logistic-derivatives-at-mbar}.
\end{proof}

The within-group proxy moments of Section~\ref{subsec:glm-within-group} are
therefore
\begin{equation}
\label{eq:logistic-proxy-moments}
  \widetilde\mu_f(s)
  =
  \frac{1}{1+e^{-m_s}},
  \qquad
  \widetilde\sigma_f(s)
  =
  \widetilde\mu_f(s)\bigl(1-\widetilde\mu_f(s)\bigr)b_s.
\end{equation}

\subsection{Leading decomposition on the probability scale}
\label{subsec:logistic-decomposition}

Substituting \eqref{eq:logistic-derivatives-at-mbar} into the general
decomposition gives the logistic specialization.

\begin{theorem}[Leading logistic decomposition]
\label{thm:logistic-decomposition}
Under the assumptions of Theorem~\ref{thm:glm-two-layer}, the output-scale
criterion satisfies
\[
  U_2(f)
  =
  D_{1,\mathrm{logit}}(f)
  +
  \Delta_{\mathrm{within}}(f)
  +
  \Delta_{\mathrm{across}}(f),
\]
where the leading logistic decomposition is
\begin{align}
D_{1,\mathrm{logit}}(f)
&=
a^2
\Big[
\underbrace{\Var_S(\gamma_S)}_{\textnormal{direct mean}}
+
\underbrace{\Var_S\bigl(\langle\bmu^{(S)},\bbeta\rangle\bigr)}
_{\textnormal{indirect mean}}
\notag\\
&\qquad\qquad
+
\underbrace{2\,\Cov_S\bigl(\gamma_S,
\langle\bmu^{(S)},\bbeta\rangle\bigr)}_{\textnormal{interaction}}
+
\underbrace{\Var_S(B)}_{\textnormal{indirect structural}}
\Big]
\notag\\
&\quad+
\underbrace{
2a^2(1-2\bar\mu)\bar b\,\Cov_S(M,B)}
_{\textnormal{curvature coupling}}
+
\underbrace{
a^2(1-2\bar\mu)^2\bar b^2\,\Var_S(M)}
_{\textnormal{curvature amplification}}.
\label{eq:logistic-leading-decomposition}
\end{align}
\end{theorem}

\begin{proof}
By Lemma~\ref{lem:logistic-derivatives},
\[
  h'(\bar m)=a,
  \qquad
  h''(\bar m)=a(1-2\bar\mu).
\]
Substituting these identities into
Definition~\ref{def:glm-leading-decomposition} and
Theorem~\ref{thm:glm-two-layer} gives the result.
\end{proof}

\begin{proposition}[Logistic link geometry]
\label{prop:logistic-geometry}
Let $a=\bar\mu(1-\bar\mu)$. Then:
\begin{enumerate}
\item $0<a\le 1/4$, with equality if and only if $\bar\mu=1/2$;
\item the curvature amplification term in
\eqref{eq:logistic-leading-decomposition} is nonnegative;
\item the sign of the curvature coupling term is the sign of
\[
  (1-2\bar\mu)\Cov_S(M,B).
\]
\end{enumerate}
\end{proposition}

\begin{proof}
The map $u\mapsto u(1-u)$ is maximized on $[0,1]$ at $u=1/2$, with value
$1/4$. This proves the first claim. The second follows from the squared factors
in the amplification term. The third follows directly from
\eqref{eq:logistic-leading-decomposition}.
\end{proof}

\subsection{Actuarial reading}
\label{subsec:logistic-interpretation}

The logistic specificity is governed by the operating point $\bar\mu$. The four
latent channels inherited from the linear decomposition are multiplied by
\[
  a^2=\{\bar\mu(1-\bar\mu)\}^2,
\]
the squared local slope of the inverse-logit map. This factor is largest when
$\bar\mu=1/2$ and decreases near the boundaries of the probability scale. For
rare events, $a\simeq \bar\mu$, so a sizeable logit-scale disparity may correspond
to a small absolute disparity in predicted probabilities.

Curvature adds two effects. At $\bar\mu=1/2$, one has $h''(\bar m)=0$, so the
leading decomposition reduces locally to a slope-rescaled version of the linear
benchmark. Away from this point, the curvature coupling term links differences
in average logit scores to differences in logit-score dispersion. Its sign depends
both on the operating point and on $\Cov_S(M,B)$. The curvature amplification
term is always nonnegative and increases when the latent group means differ.

For binary insurance or health outcomes, this distinction matters. A model may
appear to generate limited disparities on the probability scale because the event
is rare, even though the same model displays visible group differences on the
logit scale. Conversely, the practical meaning of an explicit group coefficient
depends on the baseline event probability. The logistic decomposition should
therefore be interpreted together with the fitted event rate and with the
approximation diagnostics
\[
  \Delta_{\mathrm{within}}(f),
  \qquad
  \Delta_{\mathrm{across}}(f),
  \qquad
  r_{\mathrm{within}}(f)
\]
defined in Section~\ref{subsec:glm-approximation-assessment}.

\subsection{A hospitalization-risk illustration}
\label{sec:logistic-case-study}

We illustrate the logistic decomposition using the Medical Expenditure Panel
Survey Household Component (MEPS-HC), HC-192 file, corresponding to the 2016
Full Year Consolidated Data File. The response is the indicator of having at
least one hospitalization during the year,
\[
  Y
  =
  \mathbf 1\{\texttt{IPDIS16}+\texttt{IPZERO16}>0\},
\]
where \texttt{IPDIS16} is the number of hospital discharges and
\texttt{IPZERO16} is the number of zero-night hospital stays. Details on the
sample, variables, and estimation procedure are reported in
\ref{app:logistic}.

We consider the same two sensitive comparisons as in the linear illustration:
Hispanic versus non-Hispanic individuals, and uninsured versus insured
individuals. The base logistic model includes the sensitive attribute and the
selected non-sensitive covariates. The moment-aligned model applies the
latent-scale correction of Section~\ref{subsec:glm-fair-correction} and then
maps the corrected score through the inverse-logit function.

Figure~\ref{fig:glm:logit:coeff} displays the coefficient shifts induced by this
latent-scale correction. The figure should be read on the logit scale. The
non-sensitive slopes are multiplied by the group-specific factor
$\alpha_s=\bar b/b_s$, while the original sensitive effect is absorbed into
group-specific intercepts. Since the final prediction is
$h(\eta)=(1+\exp(-\eta))^{-1}$, these coefficient shifts do not translate
additively into probability shifts. Their effect depends on the local slope and
curvature of the inverse-logit map.

\begin{figure}[!htbp]
    \centering
    \includegraphics[width=\linewidth]{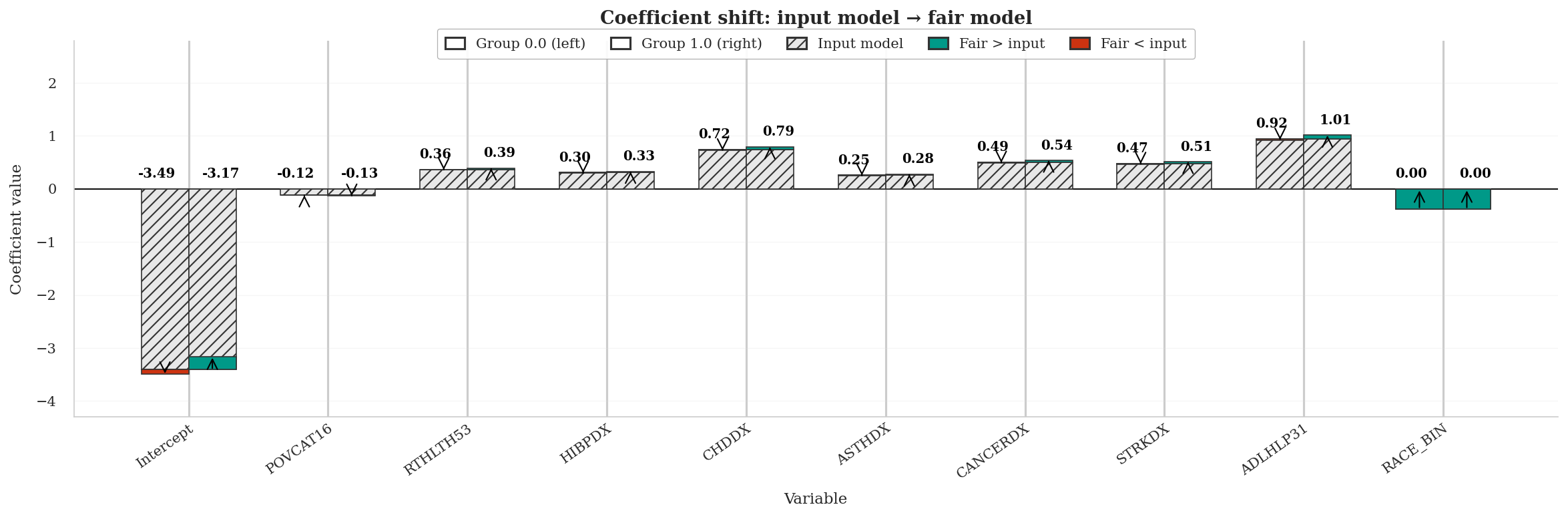}
    \includegraphics[width=\linewidth]{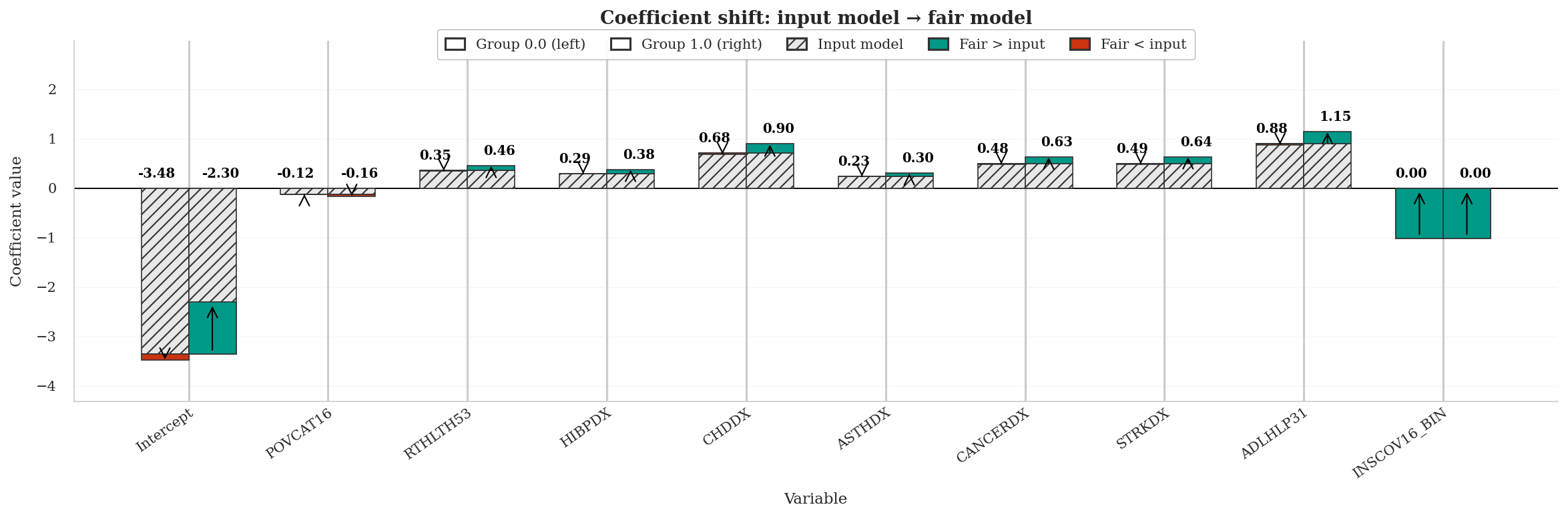}
    \caption{Logistic model: coefficient shifts from the base model to the
    latent moment-aligned model. Hatched bars show input coefficients; colored
    extensions show the post-processing correction. The sensitive-attribute
    coefficient is absorbed into group-specific intercepts. The top panel uses
    the Hispanic/non-Hispanic comparison; the bottom panel uses the
    uninsured/insured comparison.}
    \label{fig:glm:logit:coeff}
\end{figure}

Table~\ref{tab:logistic_decomposition} reports predictive performance,
output-scale moment disparity, and the leading logistic decomposition. The
table distinguishes the exact empirical criterion $U_2(f)$, the within-group
proxy $\widetilde U_2(f)$, and the leading decomposition
$D_{1,\mathrm{logit}}(f)$. The signed gap
$  \Delta_{\mathrm{within}}(f)
  =
  U_2(f)-\widetilde U_2(f)
$
measures the accuracy of the within-group delta approximation. The relative
quantity $r_{\mathrm{within}}$ is unsigned.

\begin{table}[!htbp]
\centering
\footnotesize
\begin{tabular}{lcc}
\toprule
 & \textbf{Hispanic vs.} & \textbf{Uninsured} \\
& \textbf{non-Hispanic} & \textbf{vs.\ insured} \\
\midrule
\multicolumn{3}{l}{\textit{Predictive performance}} \\
\midrule
Base model: deviance
& 11\,732 \,\scriptsize{(274)}
& 11\,675 \,\scriptsize{(275)} \\
Moment-aligned model: deviance
& 11\,788 \,\scriptsize{(275)}
& 11\,823 \,\scriptsize{(273)} \\
Base model: AUC
& .714 \,\scriptsize{(.007)}
& .723 \,\scriptsize{(.008)} \\
Moment-aligned model: AUC
& .709 \,\scriptsize{(.007)}
& .707 \,\scriptsize{(.007)} \\
\midrule
\multicolumn{3}{l}{\textit{Moment-based disparity on the probability scale \;($\times 10^{-4}$)}} \\
\midrule
Base model: $U_2(f)$
& 2.34 \,\scriptsize{(.520)}
& 6.04 \,\scriptsize{(.635)} \\
Moment-aligned model: $U_2(f^{\mathrm{fair}})$
& $<0.01$
& .034 \,\scriptsize{(.022)} \\
Base model: $\widetilde U_2(f)$
& 1.42 \,\scriptsize{(.362)}
& 5.74 \,\scriptsize{(.986)} \\
Leading decomposition: $D_{1,\mathrm{logit}}(f)$
& 1.42 \,\scriptsize{(.362)}
& 5.74 \,\scriptsize{(.986)} \\
Within gap:
$\Delta_{\mathrm{within}}=U_2(f)-\widetilde U_2(f)$
& .916
& .299 \\
Relative within gap:
$r_{\mathrm{within}}$
& 39.3\% \,\scriptsize{(6.8\%)}
& 5.5\% \,\scriptsize{(8.5\%)} \\
Relative reduction in $U_2$
& 99.8\% \,\scriptsize{(0.3\%)}
& 99.4\% \,\scriptsize{(0.3\%)} \\
\midrule
\multicolumn{3}{l}{\textit{Components of $D_{1,\mathrm{logit}}(f)$ \;($\times 10^{-4}$)}} \\
\midrule
Direct mean:
$a^2\operatorname{Var}_S(\gamma_S)$
& .954 \,\scriptsize{(.253)}
& 3.28 \,\scriptsize{(.713)} \\
Indirect mean:
$a^2\operatorname{Var}_S(\langle \bmu^{(S)},\bbeta\rangle)$
& .003 \,\scriptsize{(.003)}
& .011 \,\scriptsize{(.008)} \\
Interaction:
$2a^2\operatorname{Cov}_S(\gamma_S,\langle \bmu^{(S)},\bbeta\rangle)$
& $-$.078 \,\scriptsize{(.060)}
& .361 \,\scriptsize{(.095)} \\
Indirect structural:
$a^2\operatorname{Var}_S(B)$
& .028 \,\scriptsize{(.008)}
& .084 \,\scriptsize{(.012)} \\
Curvature coupling:
$2a^2(1\!-\!2\bar\mu)\bar b\,\operatorname{Cov}_S(M,B)$
& .187 \,\scriptsize{(.038)}
& .664 \,\scriptsize{(.064)} \\
Curvature amplification:
$a^2(1\!-\!2\bar\mu)^2\bar b^2\operatorname{Var}_S(M)$
& .330 \,\scriptsize{(.087)}
& 1.34 \,\scriptsize{(.263)} \\
\midrule
Total leading decomposition:
$D_{1,\mathrm{logit}}(f)$
& 1.42 \,\scriptsize{(.362)}
& 5.74 \,\scriptsize{(.986)} \\
\bottomrule
\end{tabular}
\caption{Logistic regression on the hospitalization indicator
$Y=\mathbf 1\{\texttt{IPDIS16}+\texttt{IPZERO16}>0\}$.
The table reports predictive performance of the base and latent
moment-aligned models, the exact empirical moment disparity $U_2(f)$ on the
probability scale, the proxy $\widetilde U_2(f)$, and the leading logistic
decomposition $D_{1,\mathrm{logit}}(f)$. All disparity and decomposition entries
are in units of $10^{-4}$. Entries are means over $30$ bootstrap replications,
with standard deviations in parentheses when reported.}
\label{tab:logistic_decomposition}
\end{table}

For the Hispanic/non-Hispanic comparison, the leading disparity is mainly driven
by the direct mean component. The indirect mean component is close to zero, and
the interaction term is slightly negative. This negative interaction means that
the explicit group contribution and the covariate-induced mean score differences
partially offset each other on the logit scale; it should not be read as a
normative correction. The two curvature terms are non-negligible, especially the
curvature amplification term, which shows that the inverse-logit transformation
contributes to the probability-scale disparity. The within approximation gap is
also sizeable, with $r_{\mathrm{within}}\simeq 39\%$. Thus, the component-level
decomposition should be read as a leading diagnostic, while the empirical value
of $U_2(f)$ remains the main output-scale measure.

For the uninsured/insured comparison, the direct mean component is again the
dominant source of disparity, but the interaction is now positive: the direct and
indirect mean channels reinforce one another. The curvature terms also add
visible probability-scale disparity. In this case, the within approximation gap is
small relative to $U_2(f)$, so the leading decomposition gives a more faithful
summary of the observed output-scale disparity.

In both comparisons, the latent moment-aligned model strongly reduces
$U_2(f)$ while producing only a modest degradation in predictive performance.
This illustrates the use of the post-processing step as a diagnostic device:
it quantifies how much of the fitted probability disparity can be removed by
aligning the first two moments of the latent score, and how much predictive
performance is lost in doing so.

\section{The Poisson Case}
\label{sec:poisson}

Poisson regression is the standard benchmark for actuarial frequency modelling.
In a pricing context, the quantity of interest is not the latent log-score itself,
but the expected claim count. The Poisson case therefore provides the simplest
unbounded mean-scale example: additive differences on the log-frequency scale
become multiplicative differences on the frequency scale.

\subsection{Poisson GLM and log-link derivatives}
\label{subsec:poisson-setup}

Let
\[
  Y\mid(\bX,S)\sim \mathrm{Poisson}\bigl(\mu(\bX,S)\bigr),
  \qquad
  \E[Y\mid\bX,S]=\Var(Y\mid\bX,S)=\mu(\bX,S).
\]
With the canonical log link,
\[
  g(\mu)=\log \mu,
\]
the latent score and the predicted frequency are
\[
  \eta
  =
  \langle \bx,\bbeta\rangle+\gamma_s+\beta_0,
  \qquad
  f(\bx,s)=\mu(\bx,s)=h(\eta),
\]
where
\begin{equation}
\label{eq:poisson-inverse-link}
  h(t)=e^t.
\end{equation}
As in Section~\ref{sec:glm-decomposition},
\[
  m_s
  =
  \E[\eta\mid S=s]
  =
  \langle \bmu^{(s)},\bbeta\rangle+\gamma_s+\beta_0,
  \qquad
  b_s
  =
  \sqrt{\bbeta^\top\bSigma^{(s)}\bbeta},
\]
and
\[
  M:=m_S,
  \qquad
  B:=b_S,
  \qquad
  \bar m:=\E_S[M],
  \qquad
  \bar b:=\E_S[B].
\]

For the exponential inverse link,
\[
  h'(t)=e^t,
  \qquad
  h''(t)=e^t.
\]
Thus, if
\[
  \bar\mu:=h(\bar m)=e^{\bar m},
\]
then
\begin{equation}
\label{eq:poisson-derivatives-at-mbar}
  h'(\bar m)=\bar\mu,
  \qquad
  h''(\bar m)=\bar\mu.
\end{equation}

The within-group proxy moments of Section~\ref{subsec:glm-within-group} are
therefore
\begin{equation}
\label{eq:poisson-proxy-moments}
  \widetilde\mu_f(s)=e^{m_s},
  \qquad
  \widetilde\sigma_f(s)=e^{m_s}b_s.
\end{equation}

\subsection{Leading decomposition on the frequency scale}
\label{subsec:poisson-decomposition}

Substituting \eqref{eq:poisson-derivatives-at-mbar} into the general GLM
decomposition gives the Poisson specialization.

\begin{theorem}[Leading Poisson decomposition]
\label{thm:poisson-decomposition}
Under the assumptions of Theorem~\ref{thm:glm-two-layer}, the output-scale
criterion satisfies
\[
  U_2(f)
  =
  D_{1,\mathrm{Pois}}(f)
  +
  \Delta_{\mathrm{within}}(f)
  +
  \Delta_{\mathrm{across}}(f),
\]
where the leading Poisson decomposition is
\begin{align}
D_{1,\mathrm{Pois}}(f)
&=
\bar\mu^2
\Big[
\underbrace{\Var_S(\gamma_S)}_{\textnormal{direct mean}}
+
\underbrace{\Var_S\bigl(\langle\bmu^{(S)},\bbeta\rangle\bigr)}
_{\textnormal{indirect mean}}
\notag\\
&\qquad\qquad
+
\underbrace{2\,\Cov_S\bigl(\gamma_S,
\langle\bmu^{(S)},\bbeta\rangle\bigr)}_{\textnormal{interaction}}
+
\underbrace{\Var_S(B)}_{\textnormal{indirect structural}}
\Big]
\notag\\
&\quad+
\underbrace{
2\bar\mu^2\bar b\,\Cov_S(M,B)}
_{\textnormal{curvature coupling}}
+
\underbrace{
\bar\mu^2\bar b^2\,\Var_S(M)}
_{\textnormal{curvature amplification}}.
\label{eq:poisson-leading-decomposition}
\end{align}
\end{theorem}

\begin{proof}
For the log link, \eqref{eq:poisson-derivatives-at-mbar} gives
\[
  h'(\bar m)=\bar\mu,
  \qquad
  h''(\bar m)=\bar\mu.
\]
Substituting these identities into
Definition~\ref{def:glm-leading-decomposition} and
Theorem~\ref{thm:glm-two-layer} gives
\eqref{eq:poisson-leading-decomposition}.
\end{proof}

The exponential link also gives a simple ratio interpretation.

\begin{proposition}[Frequency ratio induced by score differences]
\label{prop:poisson-ratio}
For any two groups $s,t\in[M]$,
\begin{equation}
\label{eq:poisson-ratio}
  \frac{\widetilde\mu_f(s)}{\widetilde\mu_f(t)}
  =
  e^{m_s-m_t}.
\end{equation}
\end{proposition}

\begin{proof}
This follows directly from
\eqref{eq:poisson-proxy-moments}.
\end{proof}

\begin{proposition}[Poisson link geometry]
\label{prop:poisson-geometry}
For the log-link Poisson model:
\begin{enumerate}
\item the curvature amplification term in
\eqref{eq:poisson-leading-decomposition} is nonnegative;
\item the sign of the curvature coupling term is the sign of
\[
  \Cov_S(M,B);
\]
\item all leading components are scaled by the square of the baseline frequency
level $\bar\mu^2$, except for the additional factors involving $\bar b$ in the
curvature terms.
\end{enumerate}
\end{proposition}

\begin{proof}
The amplification term is
\[
  \bar\mu^2\bar b^2\Var_S(M),
\]
which is nonnegative. In the coupling term
\[
  2\bar\mu^2\bar b\,\Cov_S(M,B),
\]
all factors except $\Cov_S(M,B)$ are nonnegative. The scaling statement follows
directly from \eqref{eq:poisson-leading-decomposition}.
\end{proof}

\subsection{Actuarial reading}
\label{subsec:poisson-interpretation}

The Poisson decomposition has a direct frequency interpretation. A difference
that is additive on the log-frequency scale becomes multiplicative on the
expected-count scale. Thus, a direct group effect $\gamma_s$ shifts the latent
score but changes the predicted frequency through a frequency ratio. The same
is true for indirect effects: group differences in average non-sensitive
covariates translate into multiplicative differences in expected counts.

The leading terms in \eqref{eq:poisson-leading-decomposition} are scaled by
$\bar\mu^2$. Hence, the numerical size of the frequency-scale diagnostic depends
on the baseline frequency level. In a low-frequency portfolio, a visible
log-score disparity may correspond to a small absolute difference in predicted
counts. In a portfolio with higher baseline frequency, the same log-score
disparity can have a larger mean-scale effect.

The curvature terms have a simpler sign structure than in the logistic case.
Because the exponential link is everywhere convex, the amplification term is
always nonnegative, and the coupling term has the sign of $\Cov_S(M,B)$. If
groups with larger average log-scores also have larger score dispersion, mean and
structural differences reinforce each other on the frequency scale. If the
covariance is negative, the curvature coupling term partially offsets the other
leading components.

For actuarial frequency modelling, the direct component corresponds to an
explicit sensitive effect in the log-frequency score. The indirect component
reflects differences in average rating profiles. The structural component
reflects differences in within-group score dispersion. The curvature terms
measure how the exponential link converts these latent differences into
frequency-scale disparities.

\subsection{Exposure and count-model extensions}
\label{subsec:poisson-exposure-extensions}

In actuarial applications, frequency models are often fitted with an exposure
offset:
\begin{equation}
\label{eq:poisson-offset}
  \log \mu(\bx,s)
  =
  \log \omega
  +
  \langle \bx,\bbeta\rangle
  +
  \gamma_s
  +
  \beta_0,
\end{equation}
where $\omega>0$ denotes contract duration, observation time, or insured
exposure. If exposure is fixed by design, it simply shifts the latent score by a
constant and does not create an additional between-group component. If exposure
varies across observations and is unevenly distributed across sensitive groups,
then $\log\omega$ enters the latent score. It may contribute to the mean path,
through group differences in average exposure, and to the structural path,
through group differences in exposure dispersion.

The mean-scale part of the diagnostic is unchanged for several standard
extensions of the Poisson model. Quasi-Poisson and negative-binomial models
preserve the log-link mean structure while allowing overdispersion, for example
\[
  \Var(Y\mid\bX,S)=\phi\,\mu(\bX,S)
  \text{ or }
  \Var(Y\mid\bX,S)=\mu(\bX,S)+\alpha\mu(\bX,S)^2.
\]
As long as the fitted conditional mean satisfies
\[
  \log \mu(\bx,s)=\eta,
\]
the link-induced transport in
\eqref{eq:poisson-leading-decomposition} remains the same. The response
variance function affects estimation and predictive uncertainty, but not the
leading decomposition of the fitted mean predictions.

Zero-inflated and hurdle models require a separate treatment. They introduce
at least two prediction mechanisms: one for structural zeros and one for the
positive-count intensity. In such models, group disparities may arise through
the zero-generating component, through the count intensity, or through their
combination. A full decomposition would therefore have to be multi-component,
rather than a direct application of \eqref{eq:poisson-leading-decomposition}.

\subsection{An office-visit frequency illustration}
\label{sec:poisson-case-study}

We illustrate the Poisson decomposition using the Medical Expenditure Panel
Survey Household Component (MEPS-HC), HC-192 file. The response is the number
of office-based provider visits during the year,
\[
  Y=\texttt{OBTOTV16}.
\]
The model is fitted with a log link. Details on the sample, variables, and
estimation procedure are reported in \ref{app:poisson}. We consider
the same two sensitive comparisons as before: Hispanic versus non-Hispanic
individuals, and uninsured versus insured individuals.

The base model includes the sensitive attribute and the selected non-sensitive
covariates. The moment-aligned model applies the latent-scale correction of
Section~\ref{subsec:glm-fair-correction} and then maps the corrected score
through the exponential inverse link. Thus, the correction aligns the first two
moments of the log-frequency score; its effect on predicted counts is
multiplicative.

Figure~\ref{fig:ex:poisson:beta} displays the coefficient shifts induced by the
latent-scale correction. The figure should be read on the log-frequency scale:
the non-sensitive slopes are multiplied by the group-specific factor
$\alpha_s=\bar b/b_s$, while the sensitive-attribute contribution is absorbed
into group-specific intercepts.

\begin{figure}[!htbp]
    \centering
    \includegraphics[width=\linewidth]{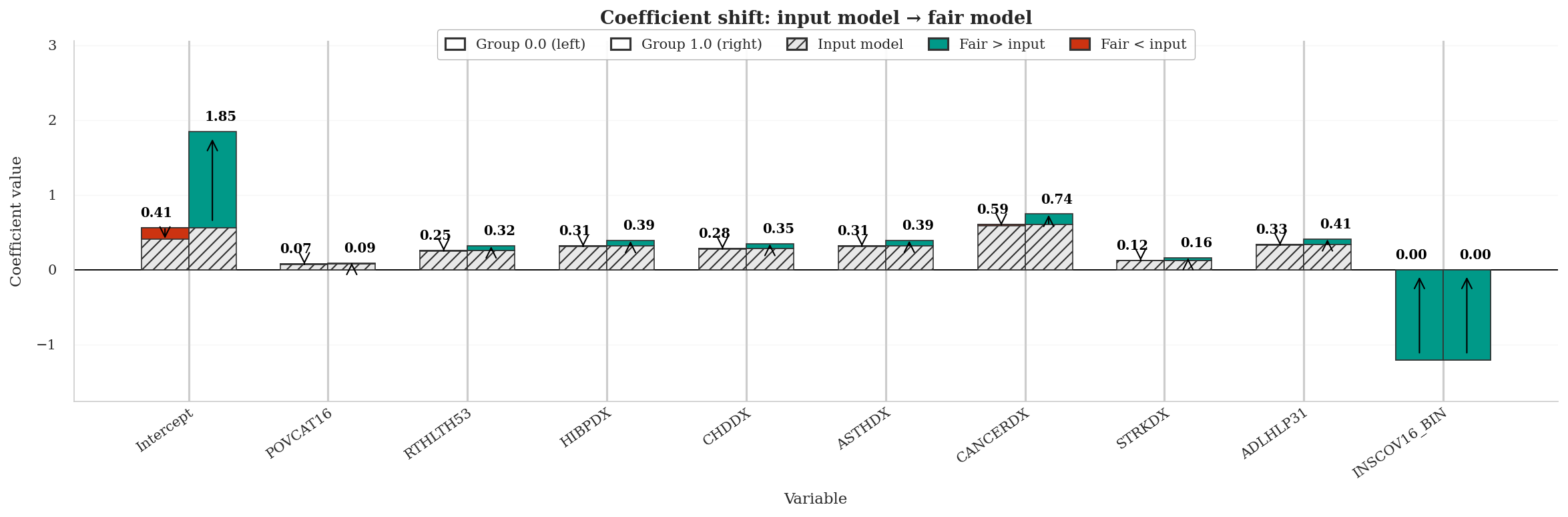}
    \includegraphics[width=\linewidth]{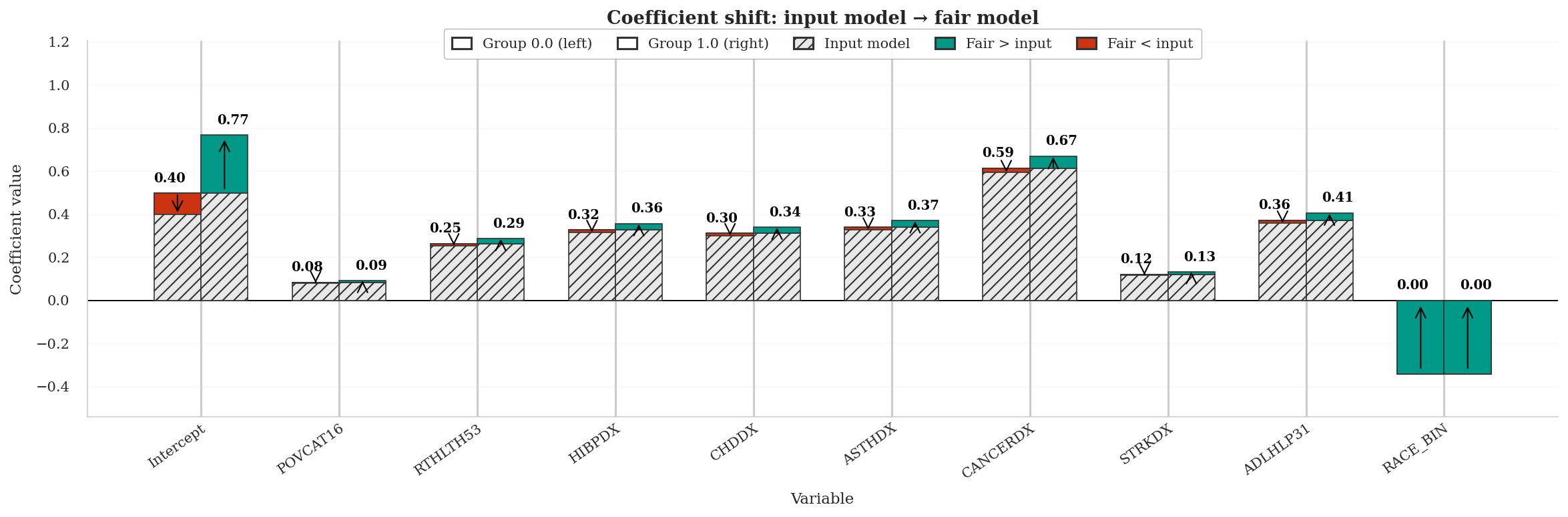}
    \caption{Poisson model: coefficient shifts from the base model to the
    latent moment-aligned model. Hatched bars show input coefficients; colored
    extensions show the post-processing correction. The sensitive-attribute
    coefficient is absorbed into group-specific intercepts. The top panel uses
    the uninsured/insured comparison; the bottom panel uses the
    Hispanic/non-Hispanic comparison.}
    \label{fig:ex:poisson:beta}
\end{figure}

Figure~\ref{fig:ex:poisson:dist} shows the conditional distributions of predicted
office-based visits before and after latent moment alignment. In the base model,
the predicted count distributions differ across sensitive groups both in location
and dispersion. After correction, the distributions are much closer, illustrating
how latent score alignment reduces moment-based disparities on the count scale.

\begin{figure}[!htbp]
    \centering
    \includegraphics[width=\linewidth]{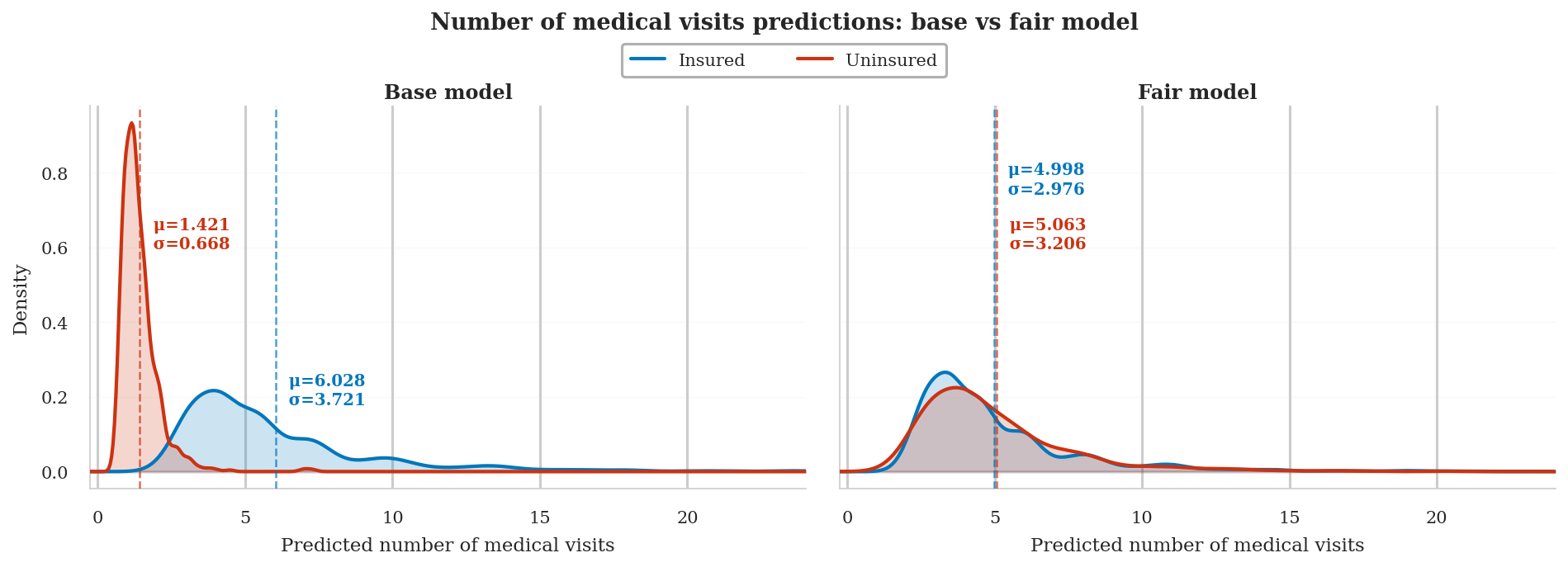}
    \includegraphics[width=\linewidth]{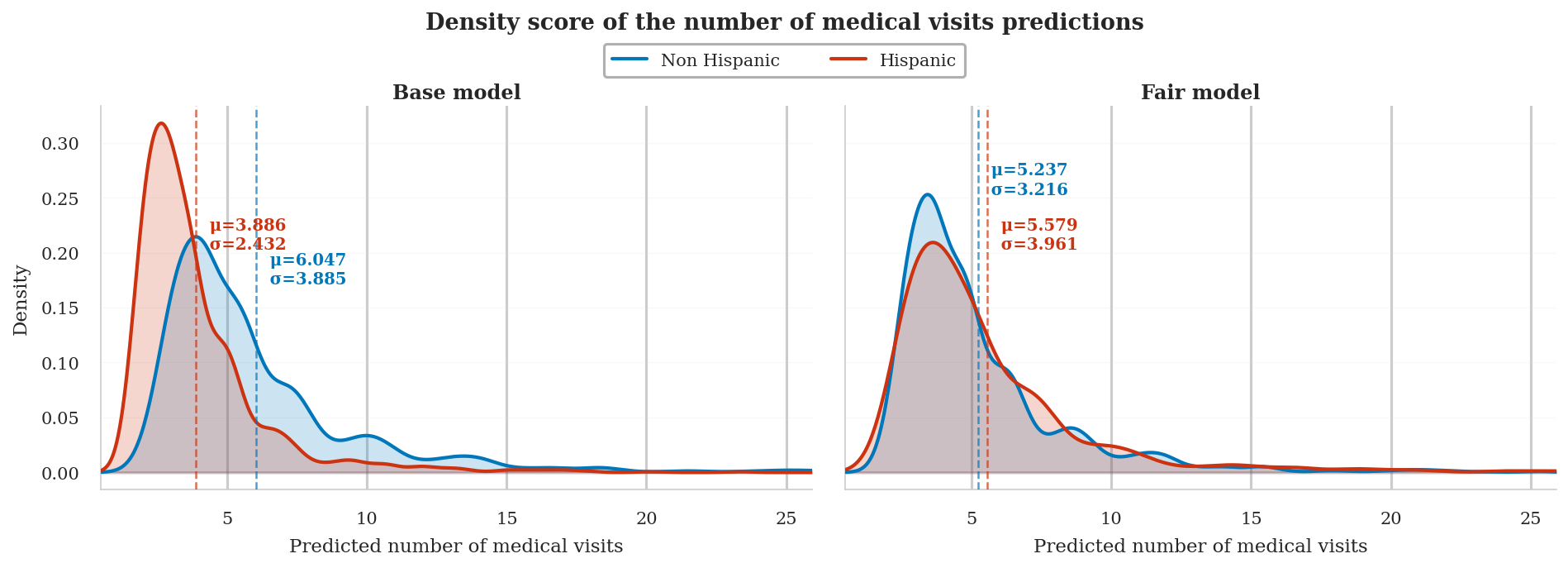}
    \caption{Poisson model: conditional density of predicted office-based visits
    (\texttt{OBTOTV16}) by sensitive group, under the base Poisson model
    (left) and the latent moment-aligned Poisson model (right). The top panel
    uses the uninsured/insured comparison; the bottom panel uses the
    Hispanic/non-Hispanic comparison.}
    \label{fig:ex:poisson:dist}
\end{figure}

Table~\ref{tab:poisson_decomposition} reports predictive performance,
moment-based disparity on the count scale, and the leading Poisson
decomposition. As in the logistic case, the table distinguishes the exact
empirical criterion $U_2(f)$, the proxy $\widetilde U_2(f)$, and the leading
decomposition $D_{1,\mathrm{Pois}}(f)$. The signed within gap
$  \Delta_{\mathrm{within}}(f)
  =
  U_2(f)-\widetilde U_2(f)$
may be positive or negative. The relative within gap
\[
  r_{\mathrm{within}}(f)
  =
  \frac{|U_2(f)-\widetilde U_2(f)|}{U_2(f)}
\]
is nonnegative by definition.

\begin{table}[!htbp]
\centering
\footnotesize
\begin{tabular}{lcc}
\toprule
 & \textbf{Hispanic vs.} & \textbf{Uninsured} \\
& \textbf{non-Hispanic} & \textbf{vs.\ insured} \\
\midrule
\multicolumn{3}{l}{\textit{Predictive performance}} \\
\midrule
Base model: deviance
& 229\,674 \,\scriptsize{(5\,122)}
& 223\,758 \,\scriptsize{(5\,145)} \\
Moment-aligned model: deviance
& 235\,199 \,\scriptsize{(5\,062)}
& 238\,880 \,\scriptsize{(5\,072)} \\
Base model: RMSE
& 10.49 \,\scriptsize{(.31)}
& 10.43 \,\scriptsize{(.31)} \\
Moment-aligned model: RMSE
& 10.51 \,\scriptsize{(.30)}
& 10.55 \,\scriptsize{(.30)} \\
\midrule
\multicolumn{3}{l}{\textit{Moment-based disparity on the count scale}} \\
\midrule
Base model: $U_2(f)$
& 2.00 \,\scriptsize{(.22)}
& 3.60 \,\scriptsize{(.15)} \\
Moment-aligned model: $U_2(f^{\mathrm{fair}})$
& .007 \,\scriptsize{(.006)}
& .001 \,\scriptsize{(.002)} \\
Base model: $\widetilde U_2(f)$
& 1.39 \,\scriptsize{(.20)}
& 5.49 \,\scriptsize{(.36)} \\
Leading decomposition: $D_{1,\mathrm{Pois}}(f)$
& 1.39 \,\scriptsize{(.20)}
& 5.49 \,\scriptsize{(.36)} \\
Within gap:
$\Delta_{\mathrm{within}}=U_2(f)-\widetilde U_2(f)$
& .60
& $-$1.89 \\
Relative within gap:
$r_{\mathrm{within}}$
& 30.3\% \,\scriptsize{(4.1\%)}
& 52.5\% \,\scriptsize{(6.6\%)} \\
Relative reduction in $U_2$
& 99.6\% \,\scriptsize{(.4\%)}
& $\approx$ 100\% \\
\midrule
\multicolumn{3}{l}{\textit{Components of $D_{1,\mathrm{Pois}}(f)$}} \\
\midrule
Direct mean:
$\bar\mu^2\operatorname{Var}_S(\gamma_S)$
& .568 \,\scriptsize{(.108)}
& 3.07 \,\scriptsize{(.238)} \\
Indirect mean:
$\bar\mu^2\operatorname{Var}_S(\langle \bmu^{(S)},\bbeta\rangle)$
& .072 \,\scriptsize{(.010)}
& .095 \,\scriptsize{(.012)} \\
Interaction:
$2\bar\mu^2\operatorname{Cov}_S(\gamma_S,\langle \bmu^{(S)},\bbeta\rangle)$
& .401 \,\scriptsize{(.051)}
& 1.08 \,\scriptsize{(.081)} \\
Indirect structural:
$\bar\mu^2\operatorname{Var}_S(B)$
& .010 \,\scriptsize{(.003)}
& .023 \,\scriptsize{(.003)} \\
Curvature coupling:
$2\bar\mu^2\bar b\,\operatorname{Cov}_S(M,B)$
& .099 \,\scriptsize{(.015)}
& .289 \,\scriptsize{(.023)} \\
Curvature amplification:
$\bar\mu^2\bar b^2\operatorname{Var}_S(M)$
& .245 \,\scriptsize{(.035)}
& .928 \,\scriptsize{(.067)} \\
\midrule
Total leading decomposition:
$D_{1,\mathrm{Pois}}(f)$
& 1.39 \,\scriptsize{(.20)}
& 5.49 \,\scriptsize{(.36)} \\
\bottomrule
\end{tabular}
\caption{Poisson regression on the number of office-based visits
(\texttt{OBTOTV16}). The table reports predictive performance of the base and
latent moment-aligned models, the exact empirical moment disparity $U_2(f)$ on
the count scale, the proxy $\widetilde U_2(f)$, and the leading Poisson
decomposition $D_{1,\mathrm{Pois}}(f)$. Entries are means over $30$ bootstrap
replications, with standard deviations in parentheses when reported.}
\label{tab:poisson_decomposition}
\end{table}

For the Hispanic/non-Hispanic comparison, the leading disparity is mainly
carried by the direct mean component, with a sizeable positive interaction. This
means that the explicit group contribution and the covariate-induced mean
log-frequency differences reinforce each other. The curvature terms are smaller
than the direct mean component but are not negligible; the exponential link
therefore adds visible count-scale disparity. The within approximation gap is
moderate, with $r_{\mathrm{within}}\simeq 30\%$, so the leading decomposition
should be interpreted alongside the full empirical value of $U_2(f)$.

For the uninsured/insured comparison, the direct mean component dominates even
more strongly, and the positive interaction is also large. The proxy $\widetilde U_2(f)$ exceeds the empirical output-scale criterion $U_2(f)$, which explains the negative signed gap
\[
  \Delta_{\mathrm{within}}=U_2(f)-\widetilde U_2(f)<0.
\]
The relative discrepancy is nevertheless positive by construction and is large, about $52.5\%$. In this case, the component-level decomposition remains informative about the main leading mechanisms, but the empirical criterion $U_2(f)$ should be treated as the primary measure of count-scale disparity.

In both comparisons, the latent moment-aligned model almost eliminates the moment-based disparity on the count scale, while the deterioration in RMSE is small. The deviance increase is more visible, especially for the uninsured/insured comparison. This illustrates the risk--disparity trade-off induced by enforcing moment alignment on the latent log-frequency score.

\section{The Tweedie Case}
\label{sec:tweedie}

Tweedie generalized linear models extend the GLM framework by combining a
mean model with a power variance function. For a Tweedie response,
\[
  Y\mid(\bX,S)\sim \mathrm{Tw}_p(\mu(\bX,S),\phi),
\]
one has
\[
  \E[Y\mid\bX,S]=\mu(\bX,S),
  \text{ and }
  \Var(Y\mid\bX,S)=\phi\,\mu(\bX,S)^p.
\]
The parameter $p$ controls how the conditional variance increases with the
mean. For $p\in(1,2)$, the Tweedie family corresponds to a compound
Poisson--Gamma model, with a mass at zero and a continuous positive component.
This makes it useful for non-Gaussian outcomes with many zeros and strong
right-skewness.

From the viewpoint of the present decomposition, the important distinction is
not the variance function itself, but the inverse link $h$. The leading
decomposition of fitted mean predictions depends on $h$, on the first two
conditional moments of the latent score, and on the fitted coefficients. The
Tweedie variance power affects the fitted model and its predictive performance,
but, under a fixed link, it does not change the algebraic form of the
link-induced transport. This is why a log-link Tweedie model has the same
leading decomposition as a log-link Poisson model, even though the two models
use different response distributions.

We first present the canonical-link Tweedie model as a theoretical benchmark,
because it isolates curvature effects specific to the Tweedie inverse link. We
then focus on the log-link specification, which is more common in applications
and is used in the numerical illustration.

\subsection{Canonical-link Tweedie model}
\label{subsec:tweedie-canonical}

Assume first that
\[
  Y\mid(\bX,S)\sim \mathrm{Tw}_p(\mu(\bX,S),\phi),
  \qquad
  p\in(1,2),
\]
with canonical link
\[
  g(\mu)=\frac{\mu^{1-p}}{1-p}.
\]
The latent score is
\[
  \eta
  =
  \langle\bx,\bbeta\rangle+\gamma_s+\beta_0
  =
  g(\mu(\bx,s)).
\]
The inverse link is therefore
\begin{equation}
\label{eq:tweedie-canonical-inverse-link}
  f(\bx,s)=\mu(\bx,s)=h(\eta),
  \qquad
  h(t)=\bigl((1-p)t\bigr)^{1/(1-p)}.
\end{equation}
Since $p>1$, the canonical domain is $t<0$. We assume in this subsection
that the fitted latent score takes values in an interval contained in
$(-\infty,0)$, so that $h$ is well defined and smooth.

As in Section~\ref{sec:glm-decomposition},
\[
  m_s
  =
  \E[\eta\mid S=s]
  =
  \langle\bmu^{(s)},\bbeta\rangle+\gamma_s+\beta_0,
  \qquad
  b_s
  =
  \sqrt{\bbeta^\top\bSigma^{(s)}\bbeta},
\]
and
\[
  M:=m_S,
  \qquad
  B:=b_S,
  \qquad
  \bar m:=\E_S[M],
  \qquad
  \bar b:=\E_S[B].
\]

\begin{lemma}[Derivatives of the canonical Tweedie inverse link]
\label{lem:tweedie-derivatives}
Let
\[
  h(t)=\bigl((1-p)t\bigr)^{1/(1-p)},
  \qquad t<0,
  \qquad p\in(1,2).
\]
Then
\[
  h'(t)=h(t)^p,
  \qquad
  h''(t)=p\,h(t)^{2p-1}.
\]
In particular, if
\[
  \bar\mu:=h(\bar m),
\]
then
\begin{equation}
\label{eq:tweedie-derivatives-at-mbar}
  h'(\bar m)=\bar\mu^p,
  \qquad
  h''(\bar m)=p\,\bar\mu^{2p-1}.
\end{equation}
\end{lemma}

\begin{proof}
Set $u(t)=(1-p)t$ and $r=(1-p)^{-1}$, so that $h(t)=u(t)^r$. Then
\[
  h'(t)
  =
  r(1-p)u(t)^{r-1}
  =
  u(t)^{r-1}.
\]
Since $r-1=p/(1-p)$, we obtain
\[
  h'(t)
  =
  u(t)^{p/(1-p)}
  =
  h(t)^p.
\]
Differentiating once more,
\[
  h''(t)
  =
  p\,h(t)^{p-1}h'(t)
  =
  p\,h(t)^{2p-1}.
\]
Evaluating at $t=\bar m$ gives
\eqref{eq:tweedie-derivatives-at-mbar}.
\end{proof}

The within-group proxy moments are
\begin{equation}
\label{eq:tweedie-canonical-proxy-moments}
  \widetilde\mu_f(s)
  =
  \bigl((1-p)m_s\bigr)^{1/(1-p)},
  \qquad
  \widetilde\sigma_f(s)
  =
  \widetilde\mu_f(s)^p b_s.
\end{equation}

\begin{theorem}[Leading decomposition for the canonical Tweedie link]
\label{thm:tweedie-canonical-decomposition}
Under the assumptions of Theorem~\ref{thm:glm-two-layer}, the output-scale
criterion satisfies
\[
  U_2(f)
  =
  D_{1,\mathrm{Tw-can}}(f)
  +
  \Delta_{\mathrm{within}}(f)
  +
  \Delta_{\mathrm{across}}(f),
\]
where the leading canonical-link Tweedie decomposition is
\begin{align}
D_{1,\mathrm{Tw-can}}(f)
&=
\bar\mu^{2p}
\Big[
\underbrace{\Var_S(\gamma_S)}_{\textnormal{direct mean}}
+
\underbrace{\Var_S\bigl(\langle\bmu^{(S)},\bbeta\rangle\bigr)}
_{\textnormal{indirect mean}}
\notag\\
&\qquad\qquad
+
\underbrace{2\,\Cov_S\bigl(\gamma_S,
\langle\bmu^{(S)},\bbeta\rangle\bigr)}_{\textnormal{interaction}}
+
\underbrace{\Var_S(B)}_{\textnormal{indirect structural}}
\Big]
\notag\\
&\quad+
\underbrace{
2p\,\bar\mu^{3p-1}\bar b\,\Cov_S(M,B)}
_{\textnormal{curvature coupling}}
+
\underbrace{
p^2\bar\mu^{4p-2}\bar b^2\,\Var_S(M)}
_{\textnormal{curvature amplification}}.
\label{eq:tweedie-canonical-leading-decomposition}
\end{align}
\end{theorem}

\begin{proof}
By Lemma~\ref{lem:tweedie-derivatives},
\[
  h'(\bar m)=\bar\mu^p,
  \qquad
  h''(\bar m)=p\bar\mu^{2p-1}.
\]
Substitution into Definition~\ref{def:glm-leading-decomposition} and
Theorem~\ref{thm:glm-two-layer} gives
\eqref{eq:tweedie-canonical-leading-decomposition}.
\end{proof}

\begin{proposition}[Canonical Tweedie link geometry]
\label{prop:tweedie-geometry}
For $p\in(1,2)$ and $\bar\mu=h(\bar m)>0$:
\begin{enumerate}
\item the curvature amplification term in
\eqref{eq:tweedie-canonical-leading-decomposition} is nonnegative;
\item the sign of the curvature coupling term is the sign of $\Cov_S(M,B)$;
\item curvature never vanishes on the admissible domain;
\item the curvature terms scale with powers $3p-1$ and $4p-2$ of
$\bar\mu$.
\end{enumerate}
\end{proposition}

\begin{proof}
The amplification term is nonnegative because
\[
  p^2\bar\mu^{4p-2}\bar b^2\Var_S(M)\ge 0.
\]
For the coupling term, all factors in
\[
  2p\,\bar\mu^{3p-1}\bar b\,\Cov_S(M,B)
\]
are nonnegative except possibly $\Cov_S(M,B)$. Since
\[
  h''(\bar m)=p\bar\mu^{2p-1}>0,
\]
curvature does not vanish on the admissible domain. The polynomial scaling
follows directly from
\eqref{eq:tweedie-canonical-leading-decomposition}.
\end{proof}

The canonical-link case shows how the Tweedie power $p$ can enter the
prediction-scale decomposition through the curvature of the inverse link. This
is useful as a theoretical benchmark. In practice, however, Tweedie GLMs are
often fitted with a log link. We now turn to that case.

\subsection{Log-link Tweedie model}
\label{subsec:tweedie-loglink}

Assume that
\[
  Y\mid(\bX,S)\sim \mathrm{Tw}_p(\mu(\bX,S),\phi),
  \qquad
  p\in(1,2),
\]
with
\[
  \log\mu(\bx,s)
  =
  \eta
  =
  \langle\bx,\bbeta\rangle+\gamma_s+\beta_0.
\]
Then
\[
  h(t)=e^t,
  \qquad
  h'(t)=e^t,
  \qquad
  h''(t)=e^t.
\]
Thus, if
\[
  \bar\mu=e^{\bar m},
\]
the within-group proxy moments are
\begin{equation}
\label{eq:tweedie-loglink-proxy-moments}
  \widetilde\mu_f(s)=e^{m_s},
  \qquad
  \widetilde\sigma_f(s)=e^{m_s}b_s.
\end{equation}

\begin{theorem}[Leading decomposition for the log-link Tweedie model]
\label{thm:tweedie-loglink-decomposition}
Under the assumptions of Theorem~\ref{thm:glm-two-layer}, the output-scale
criterion satisfies
\[
  U_2(f)
  =
  D_{1,\mathrm{Tw-log}}(f)
  +
  \Delta_{\mathrm{within}}(f)
  +
  \Delta_{\mathrm{across}}(f),
\]
where the leading log-link Tweedie decomposition is
\begin{align}
D_{1,\mathrm{Tw-log}}(f)
&=
\bar\mu^2
\Big[
\underbrace{\Var_S(\gamma_S)}_{\textnormal{direct mean}}
+
\underbrace{\Var_S\bigl(\langle\bmu^{(S)},\bbeta\rangle\bigr)}
_{\textnormal{indirect mean}}
\notag\\
&\qquad\qquad
+
\underbrace{2\,\Cov_S\bigl(\gamma_S,
\langle\bmu^{(S)},\bbeta\rangle\bigr)}_{\textnormal{interaction}}
+
\underbrace{\Var_S(B)}_{\textnormal{indirect structural}}
\Big]
\notag\\
&\quad+
\underbrace{
2\bar\mu^2\bar b\,\Cov_S(M,B)}
_{\textnormal{curvature coupling}}
+
\underbrace{
\bar\mu^2\bar b^2\,\Var_S(M)}
_{\textnormal{curvature amplification}}.
\label{eq:tweedie-loglink-leading-decomposition}
\end{align}
\end{theorem}

\begin{proof}
For the log link, $h(t)=e^t$, hence $h'(\bar m)=h''(\bar m)=\bar\mu$.
Substituting these identities into
Definition~\ref{def:glm-leading-decomposition} and
Theorem~\ref{thm:glm-two-layer} gives
\eqref{eq:tweedie-loglink-leading-decomposition}.
\end{proof}

Under a log-link Tweedie specification, the leading decomposition coincides
formally with the Poisson decomposition in
Theorem~\ref{thm:poisson-decomposition}. This does not mean that the two fitted
models are equivalent. The Poisson model imposes
\[
  \Var(Y\mid\bX,S)=\mu(\bX,S),
\]
whereas the Tweedie model uses
\[
  \Var(Y\mid\bX,S)=\phi\,\mu(\bX,S)^p.
\]
The variance function affects estimation, fitted coefficients, dispersion, and
predictive performance. But once a fitted mean predictor
\[
  f(\bx,s)=\exp\{\eta(\bx,s)\}
\]
is obtained, the leading decomposition of its fitted predictions depends on the
log link and on the latent score moments, not directly on the variance function.

\subsection{Interpretation}
\label{subsec:tweedie-interpretation}

The Tweedie case separates two roles of the model. The first role is statistical:
the response family and the power variance function determine how the model is
estimated and how uncertainty increases with the mean. The second role is
geometric: the inverse link determines how latent score disparities are
transported to the mean scale. The decomposition developed in this paper
concerns the second role.

For the canonical link, the Tweedie power $p$ enters the link geometry itself.
The local slope is $\bar\mu^p$, and the curvature terms scale as
\[
  2p\,\bar\mu^{3p-1}\bar b\,\Cov_S(M,B)
  \qquad\text{and}\qquad
  p^2\bar\mu^{4p-2}\bar b^2\Var_S(M).
\]
Thus, two fitted models with similar latent-score disparities may generate
different mean-scale diagnostics if they operate at different levels of
$\bar\mu$ or use different Tweedie powers.

For the log link, the situation is different. The mean-scale transport is
exponential, exactly as in the Poisson case. The power parameter $p$ does not
enter the leading decomposition algebraically, except through its effect on the
estimated coefficients and latent score distribution. Therefore, when comparing
Poisson and log-link Tweedie fits, differences in the decomposition should be
read as differences induced by the fitted model, not by a different inverse-link
geometry.

\subsection{A variance-power illustration on office visits}
\label{subsec:tweedie-case-study}

We illustrate the log-link Tweedie decomposition using the Medical Expenditure
Panel Survey Household Component (MEPS-HC), HC-192 file. The response is the
number of office-based provider visits during the year,
\[
  Y=\texttt{OBTOTV16},
\]
on the untruncated sample. Details on the sample, variables, and estimation
procedure are reported in \ref{app:tweedie}.

Since $\texttt{OBTOTV16}$ is a count variable, this illustration should not be
read as a generative compound Poisson--Gamma model for continuous claim
amounts. Its role is more limited: it evaluates a log-link variance-power mean
model for a zero-inflated and overdispersed utilization outcome, and checks how
the decomposition behaves when the response distribution is changed while the
inverse link remains exponential. The algebraic decomposition remains valid for
the fitted mean predictions because it depends on the inverse link and on the
latent score moments.

We consider the same two sensitive comparisons as before: Hispanic versus
non-Hispanic individuals, and uninsured versus insured individuals. The base
model includes the sensitive attribute and the selected non-sensitive covariates.
The moment-aligned model applies the latent-scale correction of
Section~\ref{subsec:glm-fair-correction} and then maps the corrected score
through the exponential inverse link.

Figure~\ref{fig:ex:tweedie:beta} displays the coefficient shifts induced by this
latent-scale correction. As in the Poisson case, the shifts are read on the
log-mean scale. The non-sensitive slopes are multiplied by the group-specific
factor $\alpha_s=\bar b/b_s$, while the sensitive-attribute contribution is
absorbed into group-specific intercepts.

\begin{figure}[!htbp]
    \centering
    \includegraphics[width=\linewidth]{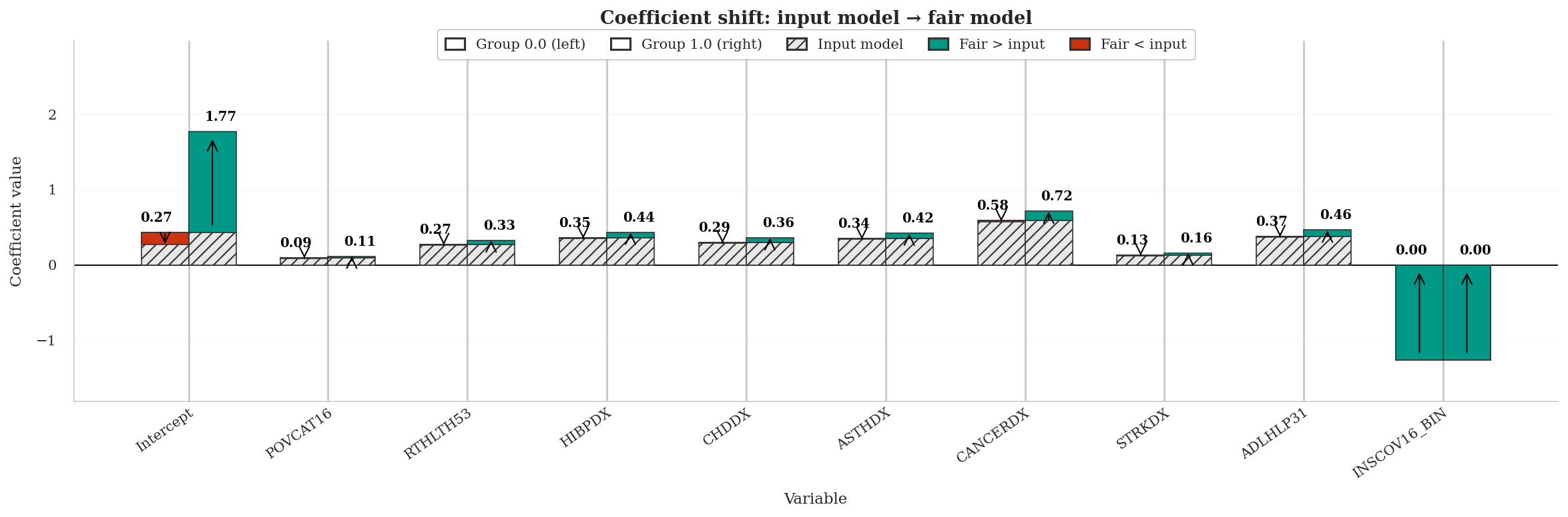}
    \caption{Tweedie log-link model: coefficient shifts from the base model to
    the latent moment-aligned model. Hatched bars show input coefficients;
    colored extensions show the post-processing correction. The correction is
    applied on the log-mean scale, before the exponential inverse link.}
    \label{fig:ex:tweedie:beta}
\end{figure}

Figure~\ref{fig:ex:tweedie:dist} displays the conditional distributions of
predicted office-based visits under the base and moment-aligned Tweedie
log-link models. The post-processing reduces the systematic group differences
visible in the base fitted predictions.

\begin{figure}[!htbp]
    \centering
    \includegraphics[width=\linewidth]{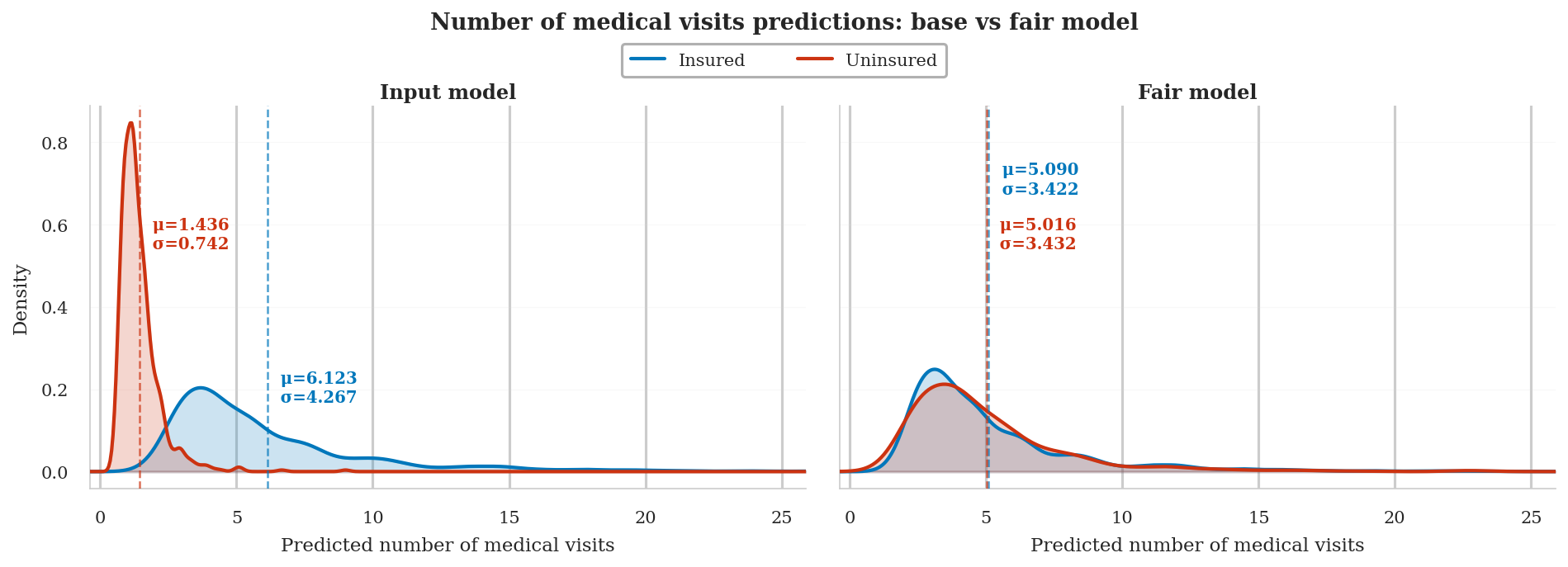}
    \caption{Tweedie log-link model: conditional density of predicted
    office-based visits (\texttt{OBTOTV16}) by sensitive group, under the base
    model and the latent moment-aligned model.}
    \label{fig:ex:tweedie:dist}
\end{figure}

Table~\ref{tab:tweedie_decomposition} reports predictive performance,
moment-based disparity on the mean scale, and the leading log-link Tweedie
decomposition. The table distinguishes the exact empirical criterion $U_2(f)$,
the proxy $\widetilde U_2(f)$, and the leading decomposition
$D_{1,\mathrm{Tw-log}}(f)$. The signed within gap
\[
  \Delta_{\mathrm{within}}(f)
  =
  U_2(f)-\widetilde U_2(f)
\]
may be positive or negative. The relative within gap
\[
  r_{\mathrm{within}}(f)
  =
  \frac{|U_2(f)-\widetilde U_2(f)|}{U_2(f)}
\]
is nonnegative by definition.

\begin{table}[!htbp]
\centering
\footnotesize
\begin{tabular}{lcc}
\toprule
& \textbf{Hispanic vs.} & \textbf{Uninsured} \\
& \textbf{non-Hispanic} & \textbf{vs.\ insured} \\
\midrule
\multicolumn{3}{l}{\textit{Predictive performance}} \\
\midrule
Base model: Tweedie deviance
& 116\,096 \,\scriptsize{(1\,278)}
& 112\,732 \,\scriptsize{(1\,347)} \\
Moment-aligned model: Tweedie deviance
& 119\,108 \,\scriptsize{(1\,268)}
& 120\,593 \,\scriptsize{(1\,276)} \\
Base model: RMSE
& 10.56 \,\scriptsize{(0.31)}
& 10.47 \,\scriptsize{(0.31)} \\
Moment-aligned model: RMSE
& 10.55 \,\scriptsize{(0.31)}
& 10.57 \,\scriptsize{(0.30)} \\
\midrule
\multicolumn{3}{l}{\textit{Moment-based disparity on the mean scale}} \\
\midrule
Base model: $U_2(f)$
& 2.83 \,\scriptsize{(0.26)}
& 4.15 \,\scriptsize{(0.19)} \\
Moment-aligned model: $U_2(f^{\mathrm{fair}})$
& .014 \,\scriptsize{(.012)}
& .002 \,\scriptsize{(.003)} \\
Base model: $\widetilde U_2(f)$
& 1.83 \,\scriptsize{(0.21)}
& 6.03 \,\scriptsize{(0.38)} \\
Leading decomposition: $D_{1,\mathrm{Tw-log}}(f)$
& 1.83 \,\scriptsize{(0.21)}
& 6.03 \,\scriptsize{(0.38)} \\
Within gap:
$\Delta_{\mathrm{within}}=U_2(f)-\widetilde U_2(f)$
& 1.00
& $-$1.88 \\
Relative within gap:
$r_{\mathrm{within}}$
& 35.4\% \,\scriptsize{(4.3\%)}
& 45.2\% \,\scriptsize{(7.1\%)} \\
Relative reduction in $U_2$
& 99.5\% \,\scriptsize{(0.5\%)}
& 100.0\% \,\scriptsize{(0.06\%)} \\
\midrule
\multicolumn{3}{l}{\textit{Components of $D_{1,\mathrm{Tw-log}}(f)$}} \\
\midrule
Direct mean:
$\bar\mu^2\operatorname{Var}_S(\gamma_S)$
& .753 \,\scriptsize{(.114)}
& 3.227 \,\scriptsize{(.239)} \\
Indirect mean:
$\bar\mu^2\operatorname{Var}_S(\langle \bmu^{(S)},\bbeta\rangle)$
& .081 \,\scriptsize{(.010)}
& .110 \,\scriptsize{(.013)} \\
Interaction:
$2\bar\mu^2\operatorname{Cov}_S(\gamma_S,\langle \bmu^{(S)},\bbeta\rangle)$
& .493 \,\scriptsize{(.050)}
& 1.18 \,\scriptsize{(.084)} \\
Indirect structural:
$\bar\mu^2\operatorname{Var}_S(B)$
& .012 \,\scriptsize{(.003)}
& .025 \,\scriptsize{(.003)} \\
Curvature coupling:
$2\bar\mu^2\bar b\,\operatorname{Cov}_S(M,B)$
& .128 \,\scriptsize{(.018)}
& .340 \,\scriptsize{(.025)} \\
Curvature amplification:
$\bar\mu^2\bar b^2\operatorname{Var}_S(M)$
& .363 \,\scriptsize{(.042)}
& 1.14 \,\scriptsize{(.079)} \\
\midrule
Total leading decomposition:
$D_{1,\mathrm{Tw-log}}(f)$
& 1.83 \,\scriptsize{(0.21)}
& 6.03 \,\scriptsize{(0.38)} \\
\bottomrule
\end{tabular}
\caption{Tweedie regression with log link on the number of office-based visits
(\texttt{OBTOTV16}). The table reports predictive performance of the base and
latent moment-aligned models, the exact empirical moment disparity $U_2(f)$
on the mean scale, the proxy $\widetilde U_2(f)$, and the leading log-link
Tweedie decomposition $D_{1,\mathrm{Tw-log}}(f)$. Entries are means over
$30$ bootstrap replications, with standard deviations in parentheses when
reported.}
\label{tab:tweedie_decomposition}
\end{table}

For the Hispanic/non-Hispanic comparison, the leading disparity is mainly
carried by the direct mean component and by the positive interaction between
the direct and indirect mean channels. The curvature amplification term is also
visible, reflecting the effect of the exponential inverse link. The within
approximation gap is moderate, with $r_{\mathrm{within}}\simeq 35\%$. Thus,
the leading components should be interpreted together with the exact empirical
criterion $U_2(f)$.

For the uninsured/insured comparison, the direct mean component is again the dominant term, and the interaction is positive. The proxy $\widetilde U_2(f)$ exceeds the empirical output-scale criterion $U_2(f)$, which explains the negative signed gap
\[
  \Delta_{\mathrm{within}}=U_2(f)-\widetilde U_2(f)<0.
\]
The relative discrepancy is nevertheless positive by construction and is substantial, about $45.2\%$. The component-level decomposition therefore identifies the main leading mechanisms, but the empirical $U_2(f)$ should remain the primary measure of mean-scale disparity.

In both comparisons, latent moment alignment almost eliminates the moment-based disparity, while producing a visible increase in Tweedie deviance and only a small change in RMSE. This is consistent with the preceding log-link examples: moment alignment substantially reduces group disparities in fitted predictions, but it should be reported together with predictive performance and approximation diagnostics.

\subsection{Extension to aggregate losses and pure-premium models}
\label{subsec:tweedie-pure-premium-opening}

The preceding illustration uses a utilization count, so it should be read as a variance-power sensitivity analysis rather than as a pure-premium application.
The same formulas, however, apply directly to aggregate loss or pure-premium modelling when the response is an annual claim amount with many zeros and a skewed positive component. In that setting, a Tweedie GLM with $p\in(1,2)$ has a natural compound Poisson--Gamma interpretation: the mass at zero corresponds to policies with no claim, while the positive component represents aggregate claim severity.

For a log-link pure-premium model,
\[
  \log \E[Y\mid\bX,S]
  =
  \langle\bx,\bbeta\rangle+\gamma_s+\beta_0,
\]
the leading decomposition remains \eqref{eq:tweedie-loglink-leading-decomposition}. Its interpretation changes: the mean-scale components are now disparities in predicted aggregate cost, or pure premium, rather than disparities in predicted utilization. The direct mean component corresponds to an explicit group contribution in the tariff score; the indirect mean component captures group differences in average rating profiles; the structural component captures differences in within-group tariff-score dispersion; and the curvature terms quantify how the exponential link converts latent tariff-score differences into premium-scale disparities.

For a canonical-link Tweedie model, the power $p$ also enters the link geometry through
\[
  h'(\bar m)=\bar\mu^p,
  \qquad
  h''(\bar m)=p\bar\mu^{2p-1}.
\]
The decomposition then depends more explicitly on both the fitted average cost level and the variance-power parameter. This makes the canonical case useful as a theoretical benchmark, whereas the log-link specification is often more natural in actuarial implementations. A pure-premium extension of the empirical application would therefore replace $\texttt{OBTOTV16}$ by an aggregate expenditure or claim amount, retain the zero observations, and report the same diagnostic quantities $U_2(f)$, $\widetilde U_2(f)$, $D_{1,\mathrm{Tw-log}}(f)$, and $r_{\mathrm{within}}$.

\section{Discussion and Conclusion}
\label{sec:conclusion}

This paper develops a moment-based decomposition framework for diagnosing
direct and indirect sources of group disparity in generalized linear models. The
motivation is actuarial: many insurance and health-risk predictions are not read
on a Gaussian or identity-link scale, but on probability, frequency, utilization,
or cost scales obtained through a nonlinear inverse link. In such settings,
linear-model intuitions are incomplete. A group difference in latent scores may
be compressed, amplified, or reshaped once it is transported to the prediction
scale.

The starting point is an exact linear-Gaussian benchmark. Under Gaussian
group-conditional score distributions, the Wasserstein barycentric measure of
distributional demographic-parity violation reduces to a two-moment criterion.
This criterion decomposes exactly into four channels: a direct mean component,
an indirect mean component, an interaction component, and an indirect structural
component. These terms describe, respectively, the explicit fitted contribution
of the sensitive attribute, the effect of group differences in average
non-sensitive covariates, the reinforcement or offset between these two mean
channels, and differences in within-group score dispersion.

The main contribution is to extend this decomposition logic to GLMs. The
extension requires separating three objects. The first is the empirical
output-scale criterion $U_2(f)$, based on the first two conditional moments of
the fitted predictions. The second is the within-group proxy
$\widetilde U_2(f)$, obtained by transporting the latent score moments through
a first-order delta approximation. The third is the leading decomposition
$D_1(f)$, obtained by expanding the proxy criterion across sensitive groups
around the global latent mean $\bar m$. This distinction makes the status of
the approximation explicit:
\[
  U_2(f)
  =
  D_1(f)
  +
  \Delta_{\mathrm{within}}(f)
  +
  \Delta_{\mathrm{across}}(f).
\]
The leading term $D_1(f)$ preserves the four linear channels and adds two
link-induced curvature components: a curvature coupling term, which links
group-wise latent means and score dispersions, and a curvature amplification
term, which converts latent mean heterogeneity into additional prediction-scale
dispersion.

The three GLM examples illustrate different aspects of this mechanism. In the
logistic case, the bounded probability scale can compress latent disparities,
especially for rare events, while the curvature changes sign around a baseline
probability of one half. In the Poisson case, additive differences on the
log-frequency scale become multiplicative differences on the count scale, and
the exponential link creates nonnegative curvature amplification. In the Tweedie
case, the distinction between the response variance function and the inverse
link becomes important. With a log link, the leading decomposition has the same
formal expression as in the Poisson case, even though the fitted model, response
distribution, and predictive performance may differ. With the canonical Tweedie
link, the variance-power parameter also enters the link geometry through the
curvature of the inverse link.

The framework also clarifies the limits of fairness through unawareness.
Removing the sensitive attribute from a model can remove the explicit direct
component, but it does not remove disparities transmitted through correlated
non-sensitive covariates, through differences in within-group covariance
structures, or through nonlinear link effects. These mechanisms have different
actuarial interpretations. A direct group coefficient, a proxy-mediated profile
effect, a structural dispersion effect, and a curvature-induced amplification do
not point to the same modelling issue and should not lead to the same
governance response.

The moment-aligning post-processing used in the numerical illustrations should
be understood in the same diagnostic spirit. It is applied on the latent score
scale, not directly on the prediction scale. By construction, it aligns the first
two conditional moments of the latent score across groups, and it eliminates the
first-order proxy disparity. In nonlinear GLMs, however, the exact output-scale
criterion $U_2(f)$ may remain positive when higher-order features of the
conditional score distributions differ across groups. Thus, the post-processing
is not a legal or normative definition of a fair tariff. It is a way to measure how
much prediction disparity is associated with the first two moments of the fitted
score, and how much predictive performance is lost when those moments are
aligned.

The proposed diagnostic has three main limitations. First, outside the
linear-Gaussian benchmark, $U_2(f)$ is not a full characterization of
distributional demographic parity. It tracks differences in conditional means and
standard deviations of fitted predictions, not equality of the full conditional
prediction laws. Second, the leading GLM decomposition is first-order. Its
accuracy depends on within-group score dispersion, separation of group-wise
latent means, and the curvature of the inverse link over the fitted score range.
For this reason, empirical applications should report $U_2(f)$,
$\widetilde U_2(f)$, $D_1(f)$, the signed approximation gaps, and the
nonnegative relative discrepancy $r_{\mathrm{within}}$. When the discrepancy
is large, the components of $D_1(f)$ remain useful directional diagnostics, but
the empirical value of $U_2(f)$ should be treated as the primary output-scale
measure. Third, the framework is statistical rather than legal. The words direct
and indirect are used to describe fitted-score mechanisms, not to establish a
legal finding of discrimination.

The diagnostic can also be adapted to conditional uses. In some insurance
applications, the relevant question is not whether predictions are equal across
sensitive groups unconditionally, but whether disparities remain after
conditioning on admissible rating information or on legally accepted risk
classes. The same decomposition can be applied within such strata, or after
constructing residualized score components relative to admissible covariates,
and then aggregated across strata. This does not remove the need for a
substantive admissibility analysis. It simply provides a way to apply the same
moment-based diagnostic to conditional comparisons.

Several extensions are natural. Multi-component insurance models, such as
zero-inflated, hurdle, or frequency--severity models, would require
decompositions that separate disparities arising from occurrence, frequency,
severity, and aggregation mechanisms. Another direction is statistical inference
for the decomposition components, for instance through bootstrap percentile
intervals or asymptotic plug-in approximations. A third direction is empirical:
the same portfolio can be analysed under occurrence, frequency, utilization,
and aggregate-cost targets to study how a disparity changes when the actuarial
decision scale changes.

The central message is that fairness diagnostics in actuarial GLMs should be
read on the scale on which predictions are used, but with a clear account of how
the model transports latent disparities to that scale. The proposed decomposition
provides such an account. It does not replace legal analysis, model validation, or actuarial judgement. It gives actuaries a structured way to ask where a prediction disparity comes from: from an explicit sensitive effect, from proxy-mediated covariate profiles, from structural differences in score dispersion, or from the nonlinear geometry of the link function.

\section*{Acknowledgments}

The authors thank seminar participants at the University of Hong Kong, ESSEC Business School, Université Claude Bernard Lyon 1, and Nanyang Technological University for helpful comments on an earlier version of this paper.

\subsection*{Declaration of conflicting interest}

The authors declared no potential conflicts of interest with respect to the research, authorship, and/or publication of this article.

\subsection*{Funding statement}

AC acknowledges funding from the Natural Sciences and Engineering Research Council of Canada (NSERC) and the SCOR Foundation for Sciences.

\section*{Data availability}

The empirical design and the computation of the diagnostic quantities are described in \ref{app:rep} and ref{app:computation}. Replication code and
notebooks will be made available in a public repository upon acceptance.

\bibliographystyle{elsarticle-harv}
\bibliography{biblio}

\appendix

\section{Empirical Design and Reproducibility}
\label{app:rep}

This appendix describes the empirical design used in
Sections~\ref{subsec:lm-gaussian-illustration},
\ref{sec:logistic-case-study},
\ref{sec:poisson-case-study}, and
\ref{subsec:tweedie-case-study}. The numerical results can be reproduced by
following the preprocessing, model-fitting, and bootstrap steps described below.
Replication notebooks will be made available online at
\texttt{[repository link to be inserted]}.

\subsection{Data source}
\label{app:data-source}

We use the Medical Expenditure Panel Survey Household Component (MEPS-HC),
file HC-192, corresponding to the 2016 Full Year Consolidated Data File. The
file contains person-level information on demographics, insurance coverage,
health status, functional limitations, chronic conditions, health-care use, and
annual expenditures. The raw file contains $34{,}655$ individuals.

\subsection{Preprocessing}
\label{app:preprocessing}

The same preprocessing principles are applied across all empirical
illustrations. We first restrict the sample to adults by keeping individuals with
\[
  \texttt{AGE16X}\geq 18.
\]
MEPS negative integer codes, corresponding to refusal, inapplicable answers, or
other missing-value categories, are recoded as missing values. Listwise deletion
is then applied separately for each modelling task, using the outcome and
predictors required for that task. Valid zeros in the outcome variables are
retained unless an explicit sample restriction is stated below.

After adult restriction and common preprocessing on the baseline variables, the
working sample contains $n=24{,}549$ individuals. The final estimation sample
may differ slightly across models because the relevant outcome variable changes.

\subsection{Predictors}
\label{app:predictors}

The four modelling tasks use the same baseline set of non-sensitive predictors.
They are chosen to capture self-perceived health, chronic conditions, and severe
functional limitation. The predictors are:

\begin{itemize}
\item general health rating, $\texttt{RTHLTH53}$;
\item mental health rating, $\texttt{MNHLTH53}$;
\item hypertension diagnosis, $\texttt{HIBPDX}$;
\item coronary heart disease diagnosis, $\texttt{CHDDX}$;
\item asthma diagnosis, $\texttt{ASTHDX}$;
\item cancer diagnosis, $\texttt{CANCERDX}$;
\item help with activities of daily living, $\texttt{ADLHLP31}$.
\end{itemize}

The variables $\texttt{RTHLTH53}$ and $\texttt{MNHLTH53}$ are originally
recorded on a five-point scale from excellent to poor. In the empirical
implementation, they are treated as categorical health-status variables. The
diagnosis and limitation variables are treated as binary indicators after
valid-code recoding.

\subsection{Sensitive variables}
\label{app:sensitive-variables}

Each empirical model is fitted twice, once for each of two binary sensitive
comparisons.

The first comparison is based on insurance coverage status. The original MEPS
variable $\texttt{INSCOV16}$ distinguishes individuals with any private
insurance, individuals with public insurance only, and individuals uninsured
for the entire year. We define
\[
  \texttt{INSCOV16\_BIN}
  =
  \mathbf 1\{\texttt{INSCOV16}=3\},
\]
so that $\texttt{INSCOV16\_BIN}=1$ denotes uninsured individuals and
$\texttt{INSCOV16\_BIN}=0$ denotes insured individuals.

The second comparison is based on ethnicity. The original variable
$\texttt{RACETHX}$ distinguishes Hispanic individuals, non-Hispanic White
individuals, non-Hispanic Black individuals, non-Hispanic Asian individuals, and
individuals reporting another or multiple racial categories. We define
\[
  \texttt{RACE\_BIN}
  =
  \mathbf 1\{\texttt{RACETHX}=1\},
\]
so that $\texttt{RACE\_BIN}=1$ denotes Hispanic individuals and
$\texttt{RACE\_BIN}=0$ denotes non-Hispanic individuals.

\subsection{Linear model}
\label{app:linear}

The linear benchmark in
Section~\ref{subsec:lm-gaussian-illustration} uses annual total health-care
expenditures. The response is
\[
  Y
  =
  \log\bigl(1+\texttt{TOTEXP16}\bigr).
\]
For this identity-link benchmark only, the regression is restricted to
individuals with positive annual expenditures, implemented as
\[
  \texttt{TOTEXP16}>1.
\]
This restriction is used to obtain a simple Gaussian-style benchmark on a
log-expenditure scale. It is not used in the count and hospitalization models.

\subsection{Logistic model}
\label{app:logistic}

The logistic illustration in Section~\ref{sec:logistic-case-study} models the
probability of having at least one hospitalization during the year. The response
is
\[
  Y
  =
  \mathbf 1\{\texttt{IPDIS16}+\texttt{IPZERO16}>0\},
\]
where $\texttt{IPDIS16}$ is the number of hospital discharges and
$\texttt{IPZERO16}$ is the number of zero-night hospital stays. The model is a
logistic regression with canonical logit link.

\subsection{Poisson model}
\label{app:poisson}

The Poisson illustration in Section~\ref{sec:poisson-case-study} models the
number of office-based provider visits during the year:
\[
  Y=\texttt{OBTOTV16}.
\]
The model is a Poisson GLM with log link. Since the MEPS file records annual
utilization over the same calendar year for all individuals, no exposure offset
is used in the reported application. Valid zero counts are retained. Negative
MEPS missing-value codes are treated as missing and are not recoded as valid
zeros.

\subsection{Tweedie model}
\label{app:tweedie}

The Tweedie illustration in Section~\ref{subsec:tweedie-case-study} uses the
same utilization outcome as the Poisson illustration,
\[
  Y=\texttt{OBTOTV16},
\]
on the untruncated valid sample. Valid zeros are retained, and no upper-tail
trimming is applied. The model is fitted with a log link and a Tweedie
variance-power specification.

Because $\texttt{OBTOTV16}$ is a count variable, this empirical illustration
is not interpreted as a generative compound Poisson--Gamma model for continuous
claim amounts. It is used as a variance-power sensitivity analysis for a
zero-inflated and overdispersed utilization outcome. The decomposition remains
well defined for the fitted mean predictions because the diagnostic depends on
the inverse link and on the latent score moments.

A direct pure-premium application would replace $\texttt{OBTOTV16}$ by an
aggregate expenditure or claim amount, retain the zero observations, and apply
the same diagnostic quantities to fitted expected annual costs.

\section{Computation of the Diagnostic Quantities}
\label{app:computation}

This appendix summarizes how the empirical criteria and decomposition
components are computed after fitting a GLM.

\subsection{Group-wise fitted quantities}
\label{app:groupwise-quantities}

For each fitted model, let
\[
  \widehat\eta_i
  =
  \langle \bx_i,\widehat\bbeta\rangle
  +
  \widehat\gamma_{S_i}
  +
  \widehat\beta_0
\]
be the fitted latent score, and let
\[
  \widehat f_i=h(\widehat\eta_i)
\]
be the fitted prediction on the mean scale. For each sensitive group $s$, we
compute
\[
  \widehat\pi_s
  =
  \frac{1}{n}\sum_{i=1}^n \mathbf 1\{S_i=s\},
\]
\[
  \widehat m_s
  =
  \frac{1}{n_s}\sum_{i:S_i=s}\widehat\eta_i,
  \qquad
  \widehat b_s
  =
  \left\{
  \frac{1}{n_s}\sum_{i:S_i=s}
  (\widehat\eta_i-\widehat m_s)^2
  \right\}^{1/2},
\]
and
\[
  \widehat\mu_f(s)
  =
  \frac{1}{n_s}\sum_{i:S_i=s}\widehat f_i,
  \qquad
  \widehat\sigma_f(s)
  =
  \left\{
  \frac{1}{n_s}\sum_{i:S_i=s}
  (\widehat f_i-\widehat\mu_f(s))^2
  \right\}^{1/2}.
\]
The empirical between-group operators $\widehat\E_S$, $\widehat\Var_S$, and
$\widehat\Cov_S$ are computed using the weights $\widehat\pi_s$.

\subsection{Exact empirical criterion}
\label{app:u2-computation}

The empirical output-scale moment disparity is
\[
  \widehat U_2(f)
  =
  \widehat\Var_S\bigl(\widehat\mu_f(S)\bigr)
  +
  \widehat\Var_S\bigl(\widehat\sigma_f(S)\bigr).
\]
This is the main empirical criterion reported in the tables.

\subsection{Within-group proxy}
\label{app:utilde-computation}

The within-group proxy moments are
\[
  \widehat{\widetilde\mu}_f(s)
  =
  h(\widehat m_s),
  \qquad
  \widehat{\widetilde\sigma}_f(s)
  =
  h'(\widehat m_s)\widehat b_s.
\]
The proxy criterion is
\[
  \widehat{\widetilde U}_2(f)
  =
  \widehat\Var_S\bigl(\widehat{\widetilde\mu}_f(S)\bigr)
  +
  \widehat\Var_S\bigl(\widehat{\widetilde\sigma}_f(S)\bigr).
\]
The signed within-group approximation gap is
\[
  \widehat\Delta_{\mathrm{within}}
  =
  \widehat U_2(f)-\widehat{\widetilde U}_2(f),
\]
and the reported relative within gap is
\[
  \widehat r_{\mathrm{within}}
  =
  \frac{
  \left|
  \widehat U_2(f)-\widehat{\widetilde U}_2(f)
  \right|}
  {\widehat U_2(f)}
\]
whenever $\widehat U_2(f)>0$. The signed gap may be positive or negative,
whereas $\widehat r_{\mathrm{within}}$ is nonnegative by definition.

\subsection{Leading decomposition}
\label{app:d1-computation}

Let
\[
  \widehat{\bar m}
  =
  \widehat\E_S[\widehat m_S],
  \qquad
  \widehat{\bar b}
  =
  \widehat\E_S[\widehat b_S].
\]
The leading decomposition is
\begin{align*}
\widehat D_1(f)
&=
\{h'(\widehat{\bar m})\}^2
\Big[
\widehat\Var_S(\widehat\gamma_S)
+
\widehat\Var_S
\bigl(
\langle \widehat\bmu^{(S)},\widehat\bbeta\rangle
\bigr)
\\
&\qquad\qquad+
2\widehat\Cov_S
\bigl(
\widehat\gamma_S,
\langle \widehat\bmu^{(S)},\widehat\bbeta\rangle
\bigr)
+
\widehat\Var_S(\widehat b_S)
\Big]
\\
&\quad+
2h'(\widehat{\bar m})h''(\widehat{\bar m})
\widehat{\bar b}\,
\widehat\Cov_S(\widehat m_S,\widehat b_S)
\\
&\quad+
\{h''(\widehat{\bar m})\}^2
\widehat{\bar b}^{\,2}
\widehat\Var_S(\widehat m_S).
\end{align*}

The six reported components are the six summands in this expression:
direct mean, indirect mean, interaction, indirect structural, curvature
coupling, and curvature amplification.

The across-group approximation gap is
\[
  \widehat\Delta_{\mathrm{across}}
  =
  \widehat{\widetilde U}_2(f)-\widehat D_1(f).
\]
In the reported binary examples, this quantity is numerically negligible after
the local expansion used in the tables. When it is not negligible, it should be
reported together with $\widehat\Delta_{\mathrm{within}}$.

\subsection{Latent moment alignment}
\label{app:moment-alignment-computation}

For each observation in group $s$, the latent moment-aligned score is
\[
  \widehat\eta_i^{\mathrm{fair}}
  =
  \widehat{\bar m}
  +
  \frac{\widehat{\bar b}}{\widehat b_s}
  \left(
  \widehat\eta_i-\widehat m_s
  \right).
\]
The corresponding fitted prediction is
\[
  \widehat f_i^{\mathrm{fair}}
  =
  h\bigl(\widehat\eta_i^{\mathrm{fair}}\bigr).
\]
All performance measures and disparity diagnostics for the moment-aligned model
are computed from the fitted values $\widehat f_i^{\mathrm{fair}}$.

\subsection{Bootstrap}
\label{app:bootstrap}

Sampling uncertainty is assessed by a nonparametric bootstrap. For each
replication $b=1,\ldots,B$, we sample $n$ observations with replacement
from the working dataset, refit the model, recompute the group-wise quantities,
and recompute all diagnostics:
\[
  U_2(f),\qquad
  \widetilde U_2(f),\qquad
  D_1(f),\qquad
  \Delta_{\mathrm{within}},\qquad
  r_{\mathrm{within}}.
\]
The reported tables use $B=30$ bootstrap replications and display bootstrap
means with standard deviations in parentheses. Percentile intervals can be
reported in supplementary material when a more detailed uncertainty assessment
is desired.

\end{document}